\documentclass{article}

\usepackage{euscript}
\usepackage{geometry}

\usepackage[utf8]{inputenc} 
\usepackage[T1]{fontenc}    
\usepackage{hyperref}       
\usepackage{url}            
\usepackage{booktabs}       
\usepackage{amsfonts}       
\usepackage{nicefrac}       
\usepackage{microtype}      

\usepackage{array}
\usepackage{multicol}
\usepackage{mathrsfs}
\usepackage{float}

\usepackage{amssymb, amsthm, epsfig}
\usepackage{amsmath}

\usepackage{graphicx}
\usepackage{epstopdf}

\usepackage{subfigure}
\usepackage{wrapfig}
\usepackage{setspace}

\usepackage{mystyle}
\usepackage{algorithm,caption}
\usepackage[noend]{algorithmic}

\usepackage{lmodern}

\theoremstyle{plain}
\newtheorem{theorem}{Theorem}[section]
\newtheorem{lemma}{Lemma}[section]
\newtheorem{corollary}{Corollary}[section]

\numberwithin{equation}{section}

\newcommand{\R}{\ensuremath{\mathbb{R}}}
\newcommand{\J}{\Eu{J}}
\newcommand{\Sh}{\Eu{S}}
\newcommand{\Sen}{\mathfrak{S}}
\newcommand{\SG}{\ensuremath{\mathfrak{S}}}
\newcommand{\T}{\Eu{T}}
\newcommand{\Tt}{\Eu{T}_\tau}
\newcommand{\dir}[1]{\ensuremath{\mathsf{d}#1}}

\newcommand{\dQ}{\ensuremath{\mathtt{d}_{Q}}}
\newcommand{\dQt}{\ensuremath{\mathtt{d}_{\tilde Q}}}

\newcommand{\dgen}{\ensuremath{\mathtt{d}}}
\newcommand{\dK}{\ensuremath{\mathtt{d}_K}}

\newcommand{\dH}{\ensuremath{\mathtt{d}_{\mathsf{H}}}}

\newcommand{\dQWt}{\ensuremath{\mathtt{d}_{\tilde Q,W}}}
\newcommand{\dist}{\ensuremath{\mathsf{dist}}}

\newcommand{\mindist}{\textsf{MinDist}\xspace}

\makeatletter
  \newcommand\figcaption{\def\@captype{figure}\caption}
  \newcommand\tabcaption{\def\@captype{table}\caption}
\makeatother

\title{\sffamily Sketched \mindist}

\author{Jeff M. Phillips \\ University of Utah \\ \texttt{jeffp@cs.utah.edu}
 \and
  Pingfan Tang \\ University of Utah \\\texttt{tang1984@cs.utah.edu}}

\begin{document}

\begin{titlepage}

\maketitle

\begin{abstract}
We sketch vectors of geometric objects $J$ through the \mindist function
\[
v_i(J) = \inf_{p \in J} \|p-q_i\|
\]
for $q_i \in Q$ from a point set $Q$.  Collecting the vector of these sketch values induces a simple, effective, and powerful distance: the Euclidean distance between these sketched vectors.
This paper shows how large this set $Q$ needs to be under a variety of shapes and scenarios.
For hyperplanes we provide direct connection to the sensitivity sampling framework, so relative error can be preserved in $d$ dimensions using $Q = O(d/\eps^2)$.  However, for other shapes, we show we need to enforce a minimum distance parameter $\rho$, and a domain size $L$.  For $d=2$ the sample size $Q$ then can be $\tilde O((L/\rho) \cdot 1/\eps^2)$.  For objects (e.g., trajectories) with at most $k$ pieces this can provide stronger \emph{for all} approximations with $\tilde O((L/\rho)\cdot k^3 / \eps^2)$ points.  Moreover, with similar size bounds and restrictions, such trajectories can be reconstructed exactly using only these sketch vectors.
\end{abstract}

\end{titlepage}

\section{Introduction}

In this paper (and a more empirically-focused companion paper~\cite{PT19a}) we introduce a new distance between geometric objects, $\dQ$.  For an object $J \in \J$, where $J \subset \R^d$,  this depends on a set of \emph{landmarks} $Q \subset \R^d$; for now let $n = |Q|$.  These landmarks induce a \emph{sketched representation} $v_Q(J) \in \R^n$ where the $i$th coordinate $v_i(J)$ is defined via a \mindist operation
\[
v_i(J) = \inf_{p \in J} \|p-q_i\|,
\]
using the $i$th landmark $q_i \in Q$.  When the object $J$ is implicit, we simply use $v_i$.
Then our new distance $\dQ$ between two objects $J_1, J_2 \in \J$ is simply the (normalized) Euclidean distance between the sketched representations
\[
\dQ(J_1, J_2) = \big\|\bar v_Q(J_1) - \bar v_Q(J_2)\big\|,
\]
where $\bar v_Q = \frac{1}{\sqrt{|Q|}} v_Q$.

Our companion paper introduces other variants of this distance (using other norms or using the $\arg\min_{p \in J}$ points on each $J \in \J$).  We focus on this version as it is the simplest, cleanest, easiest to use, and was the best or competitive with the best on all empirical tasks.  Indeed, for the pressing case of measuring a distance between trajectories, this new distance measure dominates a dozen other distance measures (including dynamic time warping, discrete Frechet distance, edit distance for real sequences) in terms of classification performance, and is considerably more efficient in clustering and nearest neighbor tasks.

The goal of this paper is to formally understand how many landmarks in $Q$ are needed for various error guarantees, and how to chose the locations of these points $Q$.

Our aims in the choice of $Q$ are two-fold:  first, we would like to approximate $\dQ$ with $\dQt$, and second we would like to recover $J \in \J$ exactly only using $v_Q(J)$.    The specific results vary depending on the initial set $Q$ and the object class $\J$.
More precisely, the approximation goal aims to preserve $d_Q$ for all objects $J$ in some class $\c{J}$ with a subset $\tilde Q \subset Q$ of landmarks.  Or possibly a weighted set of landmarks $W, \tilde Q$ with $|\tilde{Q}|=N$, so each $q_i$ is associated with a weight $w_i$ and the weighted distance is defined
\[
\dQWt(J_1, J_2) = \sqrt{\sum_{i=1}^N w_i \cdot \left(v_i(J_1) - v_i(J_2)\right)^2} = \left\| \tilde{v}_{\tilde Q}(J_1) - \tilde{v}_{\tilde Q}(J_2)\right\|.
\]
where $\tilde{v}_{\tilde{Q}}=(\tilde{v}_1,\cdots,\tilde{v}_N)$  with $\tilde{v}_i=\sqrt{w_i}v_i$.
Specifically, our aim is an \emph{$(\rho,\eps,\delta)$-approximation of $Q$ over $\J$} so when $W, \tilde Q$ is selected by a random process that succeeds with probability at least $1-\delta$, then for a \emph{pair} $J_1, J_2 \in \J$ with $\dQ(J_1, J_2) \geq \rho$
\[
(1-\eps) \dQ(J_1, J_2) \leq \dQWt(J_1, J_2) \leq (1+\eps) \dQ(J_1,J_2).
\]
When this holds for \emph{all} pairs in $\J$, we say it is a \emph{strong $(\rho,\eps,\delta)$-approximation of $Q$ over $\J$}.
In some cases we can set to $0$ either $\delta$ (the process is deterministic) or $\rho$ (this preserves even arbitrarily small distances), and may be able to use uniform weights $w_i = \frac{1}{|\tilde Q|}$ for all selected points.

\subsection{Our Results}
We begin with a special signed variant of the distance associated with the class $\J$ of $(d-1)$-dimensional hyperplanes (which for instance could model linear separators or linear regression models).  The signed variant provides $v_i(J)$ a negative value on one side of the separator.  In this variant, we show that if $Q$ is full rank, then we can recover $J$ from $v_Q(J)$, and a variant of sensitivity sampling can be used to select $O(d/(\delta \eps^2))$ points to provide a $(0,\eps,\delta)$-approximation $W,\tilde Q$.
Or by selecting $O(\frac{d}{\eps^2}(d \log d + \log \frac{1}{\delta}))$ results in a strong $O(0,\eps,\delta)$-approximation (Theorem \ref{strong coreset of Q, hyperplanes}).


Next we consider the more general case where the objects are bounded geometric objects $\Sh$.  For such objects it is useful to consider a bounded domain $\Omega_L = [0,L]^d$ (for $d$ a fixed constant), and consider the case where each $S \in \Sh$ and landmarks satisfy $S,Q \subset \Omega_L$.
In this case, the number of samples required for a $(\rho,\eps,\delta)$-approximation is
$\SG_Q \frac{1}{\eps^2 \delta}$ where
\begin{equation} \label{def of SG_Q}
\SG_Q= O\left(
\left(\frac{L}{\rho}\right)^{\frac{2d}{2+d}} \min \left( \log \frac{L}{\eta}, \log n, \left(\frac{L}{\rho}\right)^2 \right)^{\frac{2}{2+d}}
\right),
\end{equation}
where $\eta = \min_{q,q' \in Q}\|q-q'\|_\infty$.
A few special cases are worth expanding upon.
When $Q$ is continuous and uniform over $\Omega_L$ then $\SG_Q = O((L/\rho)^{\frac{2d}{2+d}})$, and this is tight in $\R^2$ at $\SG_Q = \Theta(L/\rho)$.  That is, we can show that $\SG_Q = \Theta(L/\rho)$ may be needed in general.
When $d=2$ but not necessarily uniform on $\Omega_L$, then $\SG_Q = O(\frac{L}{\rho} \min \{\sqrt{\log n}, L/\rho\})$.
And when $Q$ is on a grid over $\Omega_L$ in $\R^2$ of resolution $\Theta(\rho)$, then $\SG_Q = O(\frac{L}{\rho} \sqrt{\log \frac{L}{\rho}})$, just a $\sqrt{ \log L/\rho}$ more than the lower bound.

%

We conclude with some specific results for trajectories.  When considering the class $\T_k$ with at most $k$ segments, then  $O(\frac{1}{\eps^2} \SG_Q (k^3 \log \SG_Q + \log \frac{1}{\delta}))$ samples is sufficient for a \emph{strong} $(\rho, \eps,\delta)$-approximation.
Then when considering trajectories $\Tt$ where the critical points are at distance at least $\tau$ apart from any non-adjacent part of the curve, we can exactly reconstruct the trajectory from $v_Q$ as long as $Q$ is a grid of side length $\Omega(\tau)$.
It is much cleaner to describe the results for trajectories and $Q$ precisely on a grid, but these results should extend for any object with $k$ piecewise-linear boundaries, and critical points sufficiently separated, or $Q$ as having any point in each sufficiently dense grid cell, as opposed exactly on the grid lattice.

\subsection{Connections to other Domains, and Core Challenges}
\label{sec:connect}
Before deriving these results, it is useful to lay out the connection to related techniques, including ones that our results will build on, and the challenges in applying them.

\paragraph{Sensitivity sampling.}
Sensitivity sampling~\cite{MLLJ2010,FL11,FS12,VX12} is an important technique for our results.  This typically considers a dataset $X$ (a subset of a metric space), endowed with a measure $\mu : X \to \R^+$, and a family of cost functions $F$.  These cost functions are usually related to the fitting of a data model or a shape $S$ to $X$, and for instance on a single point $x \in X$, for $f \in F$, where
\[
f(x) = \inf_{p \in S} \|x-p\|^2
\]
is the squared distance from $x$ to the closest point $p$ on the shape $S$.  And then
$\bar f = \int _X f(x) \dir\mu(x)$.
The \emph{sensitivity}~\cite{MLLJ2010} of $x\in X$ w.r.t. $(F,X,\mu)$ is defined as:
\[
\sigma_{F,X,\mu}(x):=\sup_{f\in F}\frac{f(x)}{\bar{f}},
\]
and the \emph{total sensitivity} of $F$ is defined as:
$\Sen(F)=\int_X \sigma_{F,X,\mu}(x) \dir\mu(x)$.
This concept is quite general, and has been widely used in applications ranging from various forms of clustering~\cite{FL11,FS12} to dimensionality reduction~\cite{FSS13} to shape-fitting~\cite{VX12}.
In particular, this will allow us to draw $N$ samples $\tilde X$ iid from $X$ proportional to $\sigma_{F,X,\mu}(x)$, and weighted $\tilde w(\tilde x) = \frac{\Sen(F)}{N \cdot \sigma_{F,X,\mu}(\tilde x)}$; we call this \emph{$\sigma_{F,X,\mu}$-sensitive sampling}.
Then $\tilde X$ is a $(0,\eps,\delta)$-coreset; that is, with probability $1-\delta$ for \emph{each} $f \in F$
\[
(1-\eps) \bar f \leq \int_{\tilde X} f(\tilde x) \dir \tilde w(\tilde x) \leq (1+\eps) \bar f,
\]
using
$N = O(\frac{\Sen(F)}{\eps^2 \delta})$~\cite{MLLJ2010}.
The same error bound holds for \emph{all} $f \in F$ (then it is called a \emph{$(0,\eps,\delta)$-strong coreset}) with
$N = O(\frac{\Sen(F)}{\eps^2} (\mathfrak{s}_F \log \Sen(F) + \log \frac{1}{\delta}))$ where $\mathfrak{s}_F$ is the shattering dimension of the range space $(X, \mathsf{ranges}(F))$~\cite{VDH2016}.
Specifically, each range $r \in (X, \mathsf{ranges}(F))$ is defined as those points in a sublevel set of a specific cost function
$r = \{x \in X \mid \frac{\mu(x)}{\Sen(F)} \frac{f(x)}{\bar{f}} \leq \xi\}$
for some $f \in F$ and $\xi \in \R$.

It seems natural that a form of our results would follow directly from these approaches.  However, two significant and intertwined challenges remain.
First, our goal is to approximate the distance between a pair of sketches $\|v_Q(J_1) - v_Q(J_2)\|$, where these results effectively only preserve the norm of a single sketch $\|v_Q(J_1)\|$; this prohibits many of the geometric arguments in the prior work on this subject.
Second, the total sensitivity $\Sen(F)$ associated with unrestricted $Q$ and pairs $J_1, J_2 \in \J$ is in general unbounded (as we prove in Lemma \ref{lem:TS-lb}).
Indeed, if the total sensitivity was bounded, it would imply a mapping to bounded vector space~\cite{MLLJ2010}, wherein the subtraction of the two sketches $v_Q(J_1) - v_Q(J_2)$ would still be an element of this space, and the norm bound would be sufficient.

We circumvent these challenges in two ways.
First, we identify a special case in Section \ref{sec:H} (with negative distances, for hyperplanes) under which there is a mapping of the sketch $v_Q(J_1)$ to metric space independent of the size and structure of $Q$.  This induces a bound for total sensitivity related to a single object, and allows the subtraction of two sketches to be handled within the same framework.

Second, we enforce a lower bound on the distance $\dQ(J_1, J_2) > \rho$ and an upper bound on the domain $\Omega_L = [0,L]^d$.  This induces a restricted class of pairs $\J_{L/\rho}$ where $L/\rho$ is a scaleless parameter, and it shows up in bounds we are then able to produce for the total sensitivity with respect to $\J_{L/\rho}$ and $Q \subset \Omega_L$.


\paragraph{Leverage scores, and large scales.}
Let $(\cdot)^+$ denotes the Moore-Penrose pseudoinverse of a matrix, so $(AA^T)^+=(AA^T)^{-1}$ when $AA^T$ is full rank.
The \textit{leverage score}~\cite{DMM08} of the $i$th column $a_i$ of matrix $A$ is defined as:
$\tau_i(A):= a_i^T(AA^T)^+a_i.$
This definition is more specific and linear-algebraic than sensitivity, but has received more attention for scalable algorithm development and approximation~\cite{DMM08,boutsidis2009improved,DMMW12,CMM17,MM17,MCJ2016}.

However, Theorem \ref{theorem sensitivity} (in the Appendix \ref{sec: online row sampling}) shows that if $F$ is the collection of some functions defined on a set $Q$ of $n$ points ($\mu(q_i)=\frac{1}{n}$ for all $q_i\in Q$), where each $f\in F$ is the square of some function $v$ in a finite dimensional space $V$ spanned by a basis $\{v^{(1)},\cdots,v^{(\kappa)}\}$, then
 we can build a $\kappa \times n$ matrix $A$ where the $i$th column is $\frac{1}{\sqrt{n}}\big(v^{(1)}(q_i),\cdots,v^{(\kappa)}(q_i)\big)^T$, and have $\frac{1}{n} \cdot \sigma_{F,Q,\mu}(q_i)$ is precisely the leverage score of the $i$th column of the matrix $A$.
A similar observation has been made by Varadarajan and Xiao~\cite{VX12}.

A concrete implication of this connection is that we can invoke an online row sampling algorithm of Cohen \etal~\cite{MCJ2016}.  In our context, this algorithm would stream over $Q$, maintaining (ridge) estimates of the sensitivity of each $q_i$ from a sample $\tilde Q_{i-1}$, and retaining each $q_i$ in that sample based on this estimate.  Even in this streaming setting, this provides an approximation bound not much weaker than the sampling or gridding bounds we present; see Appendix \ref{sec: online row sampling}.

\paragraph{Connection from \mindist to shape reconstruction.}
The fields of computational topology and surface modeling have extensively explored~\cite{GeomInf,PWZ15,CCM10} the distance function to a compact set $J \subset \R^d$
\[
\dgen_J(x) = \inf_{p \in J} \|x - p\|,
\]
their approximations, and the offsets $J^r = \dgen_J^{-1}([0,r])$.  For instance the Hausdorff distance between two compact sets $J, J'$ is $\dH(J,J') = \|\dgen_J - \dgen_{J'}\|_\infty$.  The gradient of $\dgen_J$ implies stability properties about the medial axis~\cite{CL05}.
And most notably, this stability of $\dgen_J$ with respect to a sample $P \sim J$ or $P \sim \partial J$ is closely tied to the development of shape reconstruction (aka geometric and topological inference) through $\alpha$-shapes~\cite{EM94}, power crust~\cite{ACK01}, and the like.  The intuitive formulation of this problem through $\dgen_J$ (as opposed to Voronoi diagrams of $P$) has led to more statistically robust variants~\cite{CCM10,PWZ15} which also provide guarantees in shape recovery up to small feature size~\cite{CCL06}, essentially depending on the maximum curvature of $\partial J$.

Our formulation flips this around.  Instead of considering samples $P$ from $J$ (or $\partial J$) we consider samples $Q$ from some domain $\Omega \subset \R^d$.
This leads to new but similar sampling theory, still depending on some feature size (represented by various scale parameters $\rho$, $\tau$, and $\eta$), and still allowing recovery properties of the underlying objects.  While the samples $P$ from $J$ can be used to estimate Hausdorff distance via an all-pairs $O(|P|^2)$-time comparison, our formulation requires only a $O(|Q|)$-time comparison to compute $\dQ$.
We leave as open questions the recovering of topological information about an object $J \in \J$ from $v_Q(J)$.

\paragraph{Function space sketching.}
While most geometric inference sampling bounds focus on low-level geometric parameters (e.g., weak local feature size, etc), a variant based on the kernel distance $\dK(P,x)$~\cite{PWZ15} can be approximated (including useful level sets) using a uniform sample $P' \sim P$.
The kernel distance in this setting is defined $\dK(P,x) = \sqrt{1 + \mu_{K}(P) - 2\kde_P(x)}$ where the kernel density estimate is defined $\kde_P(x) = \frac{1}{|P|}\sum_{p \in P} K(p,x)$ with $K(p,x) = \exp(-\|x-p\|^2)$ and $\mu_{K}(P) = \frac{1}{P} \sum_{p \in P} \kde_P(p)$.
This sampling mechanism can be used to analyze $\kde_P$ (and thus also $\dK$) ~\cite{MFSS17} by considering a reproducing kernel Hilbert space (RKHS) $\Eu{H}_K$ associated with $K$; this is a function space so each element $\phi_K(p) = K(p,\cdot) \in \Eu{H}_K$ is a function.  And averages $\Phi_K(P) = \frac{1}{P} \sum_{p \in P} \phi_K(p) = \kde_P$ are kernel density estimates.  Ultimately, $O(\frac{1}{\eps^2} \log \frac{1}{\delta})$ samples $\tilde P$ yields~\cite{LopazPaz15} with probability $1-\delta$ that $\|\Phi_K(P) - \Phi_K(\tilde P)\|_{\Eu{H}_K} \leq \eps$ which implies $\|\kde_P - \kde_{\tilde P}\|_\infty \leq \eps$, and hence also $\|\dK(P,\cdot) - \dK(\tilde P, \cdot)\|_\infty \leq \Theta(\sqrt{\eps})$.
Notably, the natural $\Eu{H}_K$-norm is an $\ell_2$-norm when restricted to any finite dimensional subspace (e.g., the basis defined by $\{\phi_K(p)\}_{p \in P}$).

Similarly, our approximations of $\dQ(\cdot,\cdot)$ using a sample $\tilde Q \sim Q$ result in a similar function space approximation.  Again the main difference is that $\dQ$ is bivariate (so it takes in a pair $J_1, J_2 \in \J$, which is hard to interpret geometrically), and we seek a relative error (not an additive error).
This connection leads us to realize that there are JL-type approximations~\cite{JL84} of this feature space.   That is, given a set of $t$ objects $O = J_1, J_2, \ldots, J_t \subset \J$, and their representations $v_Q(J_1), v_Q(J_2), \ldots, v_Q(J_t) \in \R^n$, there is a mapping $h$ to $\R^N$ with $N = O((1/\eps^2) \log \frac{t}{\delta})$,
so with probability at least $1-\delta$ so for any pair $J,J' \in O$
$(1-\eps) \dQ(J, J') \leq \|h(v_Q(J) - h(v_Q(J')\| \leq (1+\eps) \dQ(J,J')$.
However, for such a result to hold for \emph{all} pairs in $\J$, there likely requires a lower bound on the distance $\rho$ and/or upper bound on the underlying space $L$, as with the kernels~\cite{CP17,PT19}.   Moreover, such an approach would not provide an explicit coreset $\tilde Q$ that is interpretably in the original space $\R^d$.

%

%
%
%
%
%
%
%

\section{The Distance Between Two Hyperplanes}
\label{sec:H}
\label{distance 2D}

In this section, we define a distance $\dQ$ between two hyperplanes.  Let $\c{H}=\{h \mid  h \text{ is a hyperplane  in } \R^d \}$ represent the space of all hyperplanes.

Suppose $Q=\{q_1,q_2,\cdots,q_n\}\subset \R^d$, where $q_i$ has the coordinate $(x_{i,1},x_{i,2}.\cdots,x_{i,d})$.
Without specification, in this paper $Q$ is a multiset, which means two points in $Q$ can be at the same location,
and $\|\cdot\|$ represents $l^2$ norm.

Any hyperplane $h \in \c{H}$ can be uniquely expressed in the form
\[
h=\big\{x=(x_1,\cdots,x_d)\in \R^d \ |\ \sum\nolimits_{j=1}^d u_jx_j+u_{d+1}=0 \big\},
\]
where $(u_1,\cdots,u_{d+1})$ is a vector in
$\mathbb{U}^{d+1} :=\{u=(u_1,\cdots,u_{d+1})\in \R^{d+1}  \mid \sum_{j=1}^d u_j^2=1 \text{ and the first nonzero}$ $ \text{entry of } u \text{ is positive}\}$,
i.e. $(u_1,\cdots,u_d)$ is the unit normal vector of $h$, and $u_{d+1}$ is the offset.
A sketched halfspace $h$ has $n$-dimensional vector $v_Q(h) = (v_1(h), \ldots, v_n(h))$ where each coordinate $v_i$ is defined as the \emph{signed} distance from $q_i$ to the closest points on $h$, which can be calculated
$v_i(h)=\sum_{j=1}^d u_jx_{i,j}+u_{d+1}$;
the dot-product with the unit normal of $h$, plus offset $u_{d+1}$.
As before, the distance is defined as $d_Q(h_1,h_2) = \| \frac{1}{\sqrt{n}} (v_Q(h_1) -v_Q(h_2))\|$.
%
When $Q \subset \R^d$ is \emph{full rank} -- that is, there are $d+1$ points in $Q$ which are not on a common hyperplane -- then our companion paper~\cite{PT19a} shows $d_Q$ is a metric on $\c{H}$.
\subsection{Estimation of $\dQ$ by Sensitivity Sampling on $Q$}

%

We use sensitivity sampling to estimate $\dQ$ with respect to a tuple $(F, X, \mu)$.   First suppose $Q=\{x_1,\cdots,x_n\}\subset \R^d$ is full rank and $n\geq d+1$.
Then we can let $X=Q$ and $\mu=\frac{1}{n}$; what remains is to define the appropriate $F$.
Roughly, $F$ is defined with respect to a $(d+1)$-dimensional vector space $V$, where for each $f \in F$, $f = v^2$ for some $v \in V$; and $V$ is the set of all linear functions on $x \in Q$.

We now define $F$ in more detail.
Recall each $h \in \c{H}$ can be represented as a vector $u \in \b{U}^{d+1}$.  This $u$ defines a function $v_u(q) = \sum_{i=1}^d u_i x_i + u_{d+1}$, and these functions are elements of $V$.  The vector space is however larger and defined
\[
V=\{v_a: Q\mapsto \R \mid v_a(q)=\sum_{i=1}^da_ix_i+a_{d+1} \text{ where } q=(x_1,\cdots,x_d)\in Q , a=(a_1\cdots,a_{d+1})\in \R^{d+1}\},
\]
so that there can be $v_a \in V$ for which $a \notin \b{U}^{d+1}$; rather it can more generally be in $\b{R}^{d+1}$.  Then the desired family of real-valued functions is defined
\[
F=\{f: Q\mapsto [0,\infty) \mid \exists\ v\in V \text{ s.t. } f(q)=v(q)^2,\ \forall q\in Q \}.
\]

To see how this can be applied to estimate $\dQ$, consider two hyperplanes $h_1,h_2$ in $\R^d$ and the two unique vectors $u^{(1)},u^{(2)} \in \mathbb{U}^{d+1}$ which represent them.
Now introduce the vector $u=(u_1,\cdots,u_{d+1})=u^{(1)} -u^{(2)}$; note that $u \in \b{R}^{d+1}$, but not necessarily in $\b{U}^{d+1}$.
Now for $q\in Q$ define a function $f_{h_1,h_2} \in F$ as
\[
f_{h_1,h_2}(q)=f_{h_1,h_2}(x_1,\cdots,x_d)=\big(\sum\nolimits_{i=1}^d u_ix_i+u_{d+1}\big)^2,
\]
so $\dQ(h_1,h_2)=(\frac{1}{n}\sum_{q\in Q}f_{h1,h2}(q))^\frac{1}{2}$.
And thus an estimation of $\frac{1}{n}\sum_{q\in Q}f_{h1,h2}(q)$ provides an estimation of $\dQ(h_1,h_2)$.
From Lemma \ref{lemma sensitivity}, we know the total sensitivity of $F$ is $d+1$.
In particular, given the sensitivities score $\sigma(q)$ for each $q\in Q$, we can invoke \cite{MLLJ2010}[Lemma 2.1]  to reach the following theorem.

\begin{theorem}\label{theorem estimate distance}
Consider full rank $Q \subset \R^d$ and halfspaces $\c{H}$ with $\eps,\delta \in(0,1)$.
A $\sigma$-sensitive sampling $\tilde Q$ of $(Q,F)$ of size $|\tilde Q| = \frac{d+1}{\delta, \eps^2}$ results in a $(0,\eps,\delta)$-coreset.  And thus an $(0,\eps,\delta)$-approximation so
with probability at least $1-\delta$, for each pair $h_1, h_2 \in \c{H}$
\[
	(1-\eps) \dQ(h_1,h_2)\leq  \mathtt{d}_{\tilde{Q},W}(h_1,h_2)
	\leq(1+\eps) \dQ(h_1,h_2).
\]
\end{theorem}


Now, we use the framework in  Braverman \etal~\cite{VDH2016} to construct a strong $O(0,\eps,\delta)$-approximation for $Q$ over $\c{H}$.
In the remaining part of this subsection, we assume $Q$ is a set (not a multiset), each $q\in Q$ has a weight $w(q)\in (0,1]$, and $\sum_{q\in Q}w(q)=1$.
Recall that for a range space $(Q,\c{R})$ the shattering dimension $\mathfrak{s}=\dim(Q,\c{R})$ is the smallest integer $\mathfrak{s}$ so that $|\{S \cap R \mid R \in \c{R} \}| \leq |S|^\mathfrak{s}$ for all $S \subset Q$.
We introduce ranges $\c{X}$ where each range $X_{h_1,h_2,\eta} \in \c{X}$ is defined by two halfspaces $h_1,h_2 \in \c{H}$ and a threshold $\eta > 0$.  This is defined with respect to $Q$ and a weighting $w : Q \to \b{R}_+$, specifically
\[
X_{h_1, h_2, \eta} = \{q \in Q \mid w(q) f_{h_1,h_2}(q) \leq \eta\}.
\]
Next we use the sensitivity $\sigma : Q \to \b{R}_+$ to define an adjusted range space $(Q,\c{X}')$ with adjusted weights $w'(q) = \frac{\sigma(q)}{d+1} w(q)$ and adjusted ranges $X'_{h_1,h_2,\eta} \in \c{X}'$ defined using $g_{h_1,h_2}(q) = \frac{1}{\sigma(q)} \frac{f_{h_1,h_2}(q)}{\bar f_{h_1,h_2}}$ as
\[
X'_{h_1, h_2, \eta} = \{q \in Q \mid w'(q) g_{h_1,h_2}(q) \leq \eta\}.
\]
Recall that $\bar f_{h_1,h_2} = \sum_{q \in Q} w(q) f_{h_1,h_2}(q)$.
To apply \cite{VDH2016}[Theorem 5.5] we only need to bound the shattering dimension of the adjusted range space $(Q,\c{X}')$.



\begin{lemma}\label{bound the demension for hyperplane}
	The shattering dimension of adjusted range space $(Q,\c{X}')$ is bounded by $O(d)$.
\end{lemma}

\begin{proof}
We start by rewriting any element $X'_{h_1, h_2,\eta}$ of the adjusted range space as
%
\begin{align*}
X'_{h_1, h_2, \eta}
&=
\{q\in Q \mid w'(q)g_{h_1,h_2}(x)\leq \eta \}
\\ &=
\{q\in Q \mid w(q)f_{h_1,h_2}(q)\leq \eta (d+1)\bar{f}_{h_1,h_2}\}
\\ & =
\big\{q\in Q \mid \sqrt{w(q)}\big(\sum\nolimits_{i=1}^d u_ix_i+u_{d+1})\big)\leq \big(\eta (d+1)\bar{f}_{h_1,h_2}\big)^{\frac{1}{2}}\big\}
\\ & \phantom{ == }\cap
\big\{q\in Q \mid -\sqrt{w(q)}\big(\sum\nolimits_{i=1}^d u_ix_i+u_{d+1})\big)\leq \big(\eta (d+1)\bar{f}_{h_1,h_2}\big)^{\frac{1}{2}}\big\},
\end{align*}
where $(x_1,\cdots,x_d)$ is the coordinates of $q\in Q$.
This means each set $X'_{h_1, h_2, \eta}  \in \c{X}'$ can be decomposed as the intersection of sets in two ranges over $Q$ from:
    \begin{align*}
    \begin{split}
	\c{R}_1 &=\Big\{\big\{q\in Q \mid \sqrt{w(q)}\big(\sum\nolimits_{i=1}^d u_ix_i+u_{d+1})\big)
     \leq \big(\eta (d+1)\bar{f}_{h_1,h_2}\big)^{\frac{1}{2}}\big\}|\ h_1,h_2\in \c{H},\eta \geq 0\Big\},
     \\
    \c{R}_2& =\Big\{\big\{q\in Q \mid -\sqrt{w(q)}\big(\sum\nolimits_{i=1}^d u_ix_i+u_{d+1})\big)
    \leq \big(\eta (d+1)\bar{f}_{h_1,h_2}\big)^{\frac{1}{2}}\big\}|\ h_1,h_2\in \c{H},\eta \geq 0\Big\}.
    \end{split}
    \end{align*}
    By Lemma \ref{bound the dimension of range space intersection and union},
	we only need to bound the dimension of each associated range space $(Q,\c{R}_1)$ and $(Q,\c{R}_2)$.
 	We introduce new variables $c_0\in \R, z=(z_1,\cdots,z_{d+1}),c=(c_1,\cdots,c_{d+1})\in \R^{d+1}$:
    \begin{equation*}
    \begin{split}
    z_i=&\sqrt{w(q)}x_i \ \text{ for } i\in[d], \ \ z_{d+1}=\sqrt{w(q)},\\
    c_i=&u_i \ \text{ for } i\in [d+1], \ \ c_0=-\big(r(d+1)\bar{f}_{h_1,h_2}\big)^{\frac{1}{2}}.
    \end{split}
    \end{equation*}
    Since $Q$ is a fixed set, we know $z$ only depends  on $q$, and $c_0$, $c$ only depend  on $h_1,h_2$ and $\eta$.
	By introducing new variables we construct an injective map $\varphi: Q\mapsto \R^{d+1}$, s.t. $\varphi(q)=z$. So, there is also an injective
    map	from $\c{R}_1$ to $\big\{\{z\in \varphi(Q)|\ c_0+ \langle z, c \rangle \leq 0\}|\ c_0\in \R, c\in \R^{d+1}\big\}$. Since the shattering dimension of the range space
    $(\R^{d+1},\c{H}^{d+1})$, where	$\c{H}^{d+1}=\{h \text{ is a halfspace in } \R^{d+1}\}$, is $O(d)$,
    we have $\dim(Q,\c{R}_1)=O(d)$, and similarly $\dim(Q,\c{R}_2)=O(d)$. Thus, we obtain an $O(d)$ bound for the shattering dimension of  $(Q,\c{X})$.
\end{proof}

From Lemma \ref{bound the demension for hyperplane} and \cite{VDH2016}[Theorem  5.5]  we can directly obtain a strong  $O(0,\eps,\delta)$-approximation for $Q$ over $\c{H}$.

\begin{theorem}\label{strong coreset of Q, hyperplanes}
Consider full rank $Q \subset \R^d$ and halfspaces $\c{H}$ with $\eps,\delta \in(0,1)$.
A $\sigma$-sensitive sampling $\tilde Q$ of $(Q,F)$ of size $|\tilde Q| = O(\frac{d}{\eps^2} (d \log d + \log \frac{1}{\delta}))$ results in a strong $(0,\eps,\delta)$-coreset.  And thus a strong $(0,\eps,\delta)$-approximation so
with probability at least $1-\delta$, for all $h_1, h_2 \in \c{H}$
\[
	(1-\eps) \dQ(h_1,h_2)\leq  \mathtt{d}_{\tilde{Q},W}(h_1,h_2)
	\leq(1+\eps) \dQ(h_1,h_2).
\]
\end{theorem}

\section{Sketched MinDist for Two Geometric Objects}
\label{Distance Between Two Geometric Objects}

In this section, we mildly restrict $\dQ$ to the distance between any two geometric objects, in particularly bounded closed sets.
Let $\Sh=\{S\subset \R^d \mid S \text{ is a bounded closed set}\}$ be the space of objects $\J$ we consider.

As before define $v_i(S) = \inf_{p \in S} \|p - q_i\|$, and then for $S_1,S_2\in \Sh$ define $f_{S_1, S_2}(q_i) = (v_i(S_1) - v_i(S_2))^2$.  The associated function space is $F(\Sh)=\{f_{S_1,S_2} \mid  S_1,S_2\in \Sh\}$.
Setting $\mu(q)=\frac{1}{n}$ for all $q\in Q$, then $(\dQ(S_1,S_2))^2=\bar{f}_{S_1,S_2}:=\sum_{i=1}^n \mu(q_i)f_{S_1,S_2}(q_i)$.
Using sensitivity sampling to estimate $\dQ(S_1,S_2)$ requires a bound on the total sensitivity of $F(\Sh)$.

In this section we show that while unfortunately the total sensitivity $\mathfrak{S}(F(\Sh))$ is unbounded in general, it can be tied closely to the ratio $L/\rho$ between the diameter of the domain $L$, and the minimum allowed $\dQ$ distance between objects $\rho$.  In particular, it can be at least proportional to this, and in $\R^2$ in most cases (e.g., for near-uniform $Q$) is at most proportional to $L/\rho$ or not much larger for any $Q$.

\subsection{Lower Bound on Total Sensitivity}

\begin{figure}[h]
	\centering
	\includegraphics[width=0.3\linewidth]{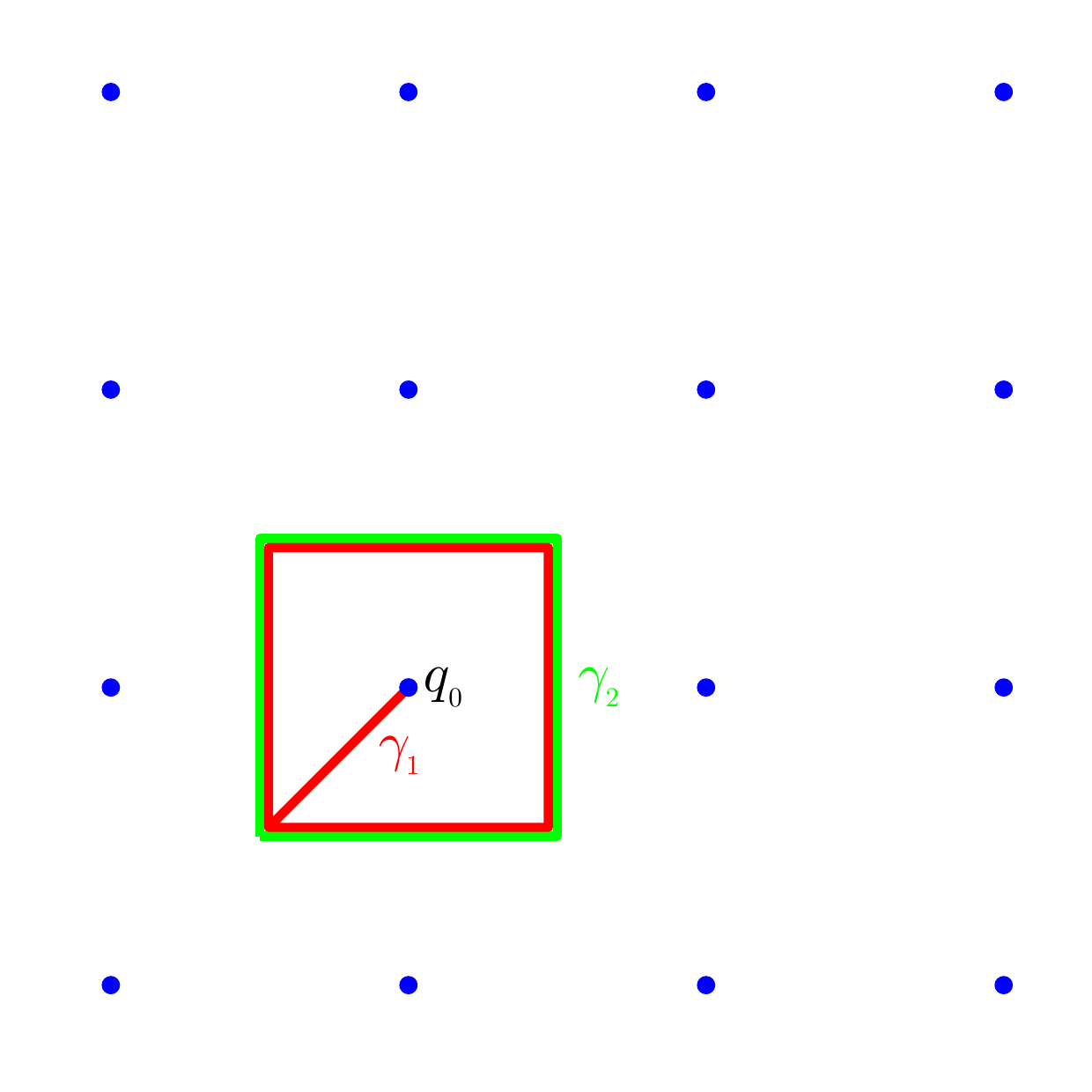}
	\vspace{-0.1in}
	\caption{\small $Q$ is the set of blue points, $\gamma_1$ is the red curve, $\gamma_2$ is the green curve, and they coincide with each other on the boundary of the square.}
    \label{fig:sens-LB}
\end{figure}

Suppose $Q$ is a set of $n$ points in $\R^2$ and no two points are at the same location, then for any $q_0\in Q$ we can draw two curves $\gamma_1,\gamma_2$ as shown in Figure \ref{fig:sens-LB}, where $\gamma_1$ is composed by five line segments and $\gamma_2$ is composed by four line segments. The four line segments of the $\gamma_2$ forms a square, on its boundary $\gamma_1$ and $\gamma_2$ coincide with each other, and inside this square, $q_0$ is the endpoint of $\gamma_1$.  We can make this square small enough, such that all points $q\neq q_0$ are outside this square. So, we have $\dist(q_0,\gamma_1)=0$ and $\dist(q_0,\gamma_2)\neq 0$, and $\dist(q,\gamma_1)=\dist(q,\gamma_2)=0$ for all
$q\neq q_0$. Thus, we have $f_{\gamma_1,\gamma_2}(q_0)>0$ and $f_{\gamma_1,\gamma_2}(q)=0$ for all $q \neq q_0$, which implies
\[
\sigma_{F(\Sh),Q,\mu}(q_0)\geq \frac{ f_{\gamma_1,\gamma_2}(q_0)}{ \bar{f}_{\gamma_1,\gamma_2}}
=\frac{f_{\gamma_1,\gamma_2}(q_0)}{\frac{1}{n}\sum_{q\in Q}f_{\gamma_1,\gamma_2}(q)}
=\frac{n f_{\gamma_1,\gamma_2}(q_0)}{f_{\gamma_1,\gamma_2}(q_0)}=n.
\]

Since this construction of two curves $\gamma_1, \gamma_2$ can be repeated around any point $q \in Q$,
\[
\mathfrak{S}(F(\Sh))=\sum_{q\in Q} \mu(q)\sigma_{F(\Sh),Q,\mu}(q)
\geq \sum_{q\in Q} \frac{1}{n}n=n.
\]


We can refine this bound by introducing two parameters $L,\rho$ for $\Sh$.
Given  $L>\rho>0$ and a set $Q\subset \R^d$ of $n$ points, we define $\Sh(L)=\{S\in \Sh \mid S\subset [0,L]^d\}$ and $F(\Sh(L),\rho)=\{f_{S_1,S_2}\in F(\Sh) \mid S_1,S_2\in \Sh(L),\ \dQ(S_1,S_2)\geq \rho\}$. The following lemma gives a lower bound for the total sensitivity of $F(\Sh(L),\rho)$ in the case $d=2$, which directly holds for larger $d$.

\begin{lemma} \label{lem:TS-lb}
	Suppose $d=2$, then can construct a set $Q\subset [0,L]^d$ such that $\mathfrak{S}(F(\Sh (L),\rho))=\Omega(\frac{L}{\rho})$.
\end{lemma}

\begin{proof}
We uniformly partition $[0,L]^2$ into $n$ grid cells, such that $C_1 \frac{L}{\rho}\leq n \leq C_2\frac{L}{\rho}$ for constants $C_1,C_2\in(0,1)$.  The side length of each grid is $\eta=\frac{L}{\sqrt{n}}$. We take $Q$ as the $n$ grid points, and for each
point $q\in Q$ we can choose two curves $\gamma_1$ and $\gamma_2$ (similar to curves in Figure \ref{fig:sens-LB}) such that $\dist(q,\gamma_1)=0$, $\dist(q,\gamma_2)\geq C_2 \eta$, and $\dist(q',\gamma_1)=\dist(q',\gamma_2)=0$ for all $q'\in Q\setminus\{q\}$ . Thus, we have $\dQ(\gamma_1,\gamma_2)\geq C_2\frac{\eta}{\sqrt{n}}=C_2\frac{L}{n}\geq \rho$. So, $f_{\gamma_1,\gamma_2}\in F(\Sh (L),\rho))$ and we have $\sigma(q)\geq n$ for all $q\in Q$ and $\mathfrak{S}(F(\Sh (L),\rho)) \geq n\geq C_1 \frac{L}{\rho}$, which implies $\mathfrak{S}(F(\Sh (L),\rho))=\Omega(\frac{L}{\rho})$.
\end{proof}

\subsection{Upper Bound on the Total Sensitivity}

A simple upper bound of $\mathfrak{S}(F(\Sh (L),\rho)$ is $O\big(\frac{L^2}{\rho^2}\big)$ follows from the $L/\rho$ constraint.
The sensitivity of each point $q\in Q$ is defined as $\sup_{f_{S_1,S_2}\in F(\Sh(L),\rho)} \frac{f_{S_1,S_2}(q)}{ \bar{f}_{S_1,S_2}}$, where $f_{S_1,S_2}(q)=O(L^2)$ for all $S_1,S_2 \in \Sh(L)$ and $q\in Q \subset [0,L]^d$, and the denominator $\bar{f}_{S_1,S_2}\geq \rho^2$
by assumption for all $f_{S_1,S_2} \in F(\Sh (L),\rho)$.  Hence, the sensitivity of each point in $Q$ is $O\big(\frac{L^2}{\rho^2}\big)$, and thus their average, the total sensitivity is $O\big(\frac{L^2}{\rho^2}\big)$ .
In this section we will improve and refine this bound.

We introduce two variables only depends on  $Q=\{q_1,\cdots,q_n\}\subset [0,L]^d$:
\begin{equation} \label{def of C_1 and C_Q}
C_q:=\max_{0<r\leq L}\frac{r^d}{L^d}\frac{n}{|Q\cap B_\infty(q,r)|} \ \ \text{ for } q\in Q,
\text{ and } C_Q:=\frac{1}{n}\sum_{q\in Q} C_q^{\frac{2}{2+d}}.
\end{equation}
where $B_\infty(q,r):=\{x\in \R^d \mid \|x-q\|_\infty\leq r\}$. Intuitively, $\frac{|Q\cap B_\infty(q,r)|}{r^d}$ is proportional to the point density
in region $B_\infty(q,r)$, and the value of $\frac{r^d}{L^d}\frac{n}{|Q\cap B_\infty(q,r)|}$ can be maximized, when the region $B_\infty(q,r)$ has smallest point density, which means
$r$ should be as large as possible but the number of points contained in  $B_\infty(q,r)$ should be as small as possible. A trivial bound of $C_q$ is $n$, but if we make $C_{q_0}=n$ for one point $q_0$, then it implies the value of $C_q$ for other points will be small, so for $C_Q$ it is possible to obtain a bound better than $n^{\frac{2}{d+2}}$.

Importantly, these quantities $C_q$ and $C_Q$ will be directly related to the sensitivity of a single point $\sigma(q)$ and the total sensitivity of the point set $\SG_Q$, respectively.  We formalize this connection in the next lemma, which for instance implies that for $d$ constant then $\SG_Q = O(C_Q \cdot (L/\rho)^{\frac{2d}{2+d}})$.

\begin{lemma} \label{sensitivity of one point}
For function family $F(\Sh(L),\rho)$
the sensitivity for any $q \in Q \in [0,L]^d$ is bounded
\[
\sigma(q)
 \leq
C_d C_q^{\frac{2}{2+d}}\Big(\frac{L}{\rho}\Big)^{\frac{2d}{2+d}},
\]
where $C_d= 4^\frac{2}{2+d}(8\sqrt{d})^{\frac{2d}{2+d}}$ and $C_q$ given by \eqref{def of C_1 and C_Q}.
\end{lemma}

\begin{proof}	
Recall $\sigma(q)=\sup_{f_{S_1,S_2}\in F(\Sh(L),\rho)}
\frac{f_{S_1,S_2}(q)}{\frac{1}{n}\sum_{q'\in Q} f_{S_1,S_2}(q')}$.
For any fixed $q \in Q$, for now suppose $f_{S_1,S_2}\in F(\Sh(L),\rho)$ satisfies this supremum $\sigma(q)=\frac{f_{S_1,S_2}(q)}{\frac{1}{n}\sum_{q'\in Q} f_{S_1,S_2}(q')}$.
We define $\dist(q,S) = \inf_{p \in S} \|q-p\|$ (so for $q_i \in Q$ then $\dist(q_i,S) = v_i(S)$), and then use the parameter $M:= |\dist(q,S_1) - \dist(q,S_2)|$, where $M^2 = f_{S_1,S_2}(q)$.
If $M=0$, then obviously $f_{S_1,S_2}(q)=M^2=0$, and $\sigma(q)=0$.
So, without loss of generality, we assume $M>0$ and $\dist(q,S_1)=\tau$ and $\dist(q,S_2) = \tau + M$.
We first prove $\sigma(q)\leq C_d C_{q} \frac{L^d}{M^d}$.
There are two cases for the relationship between $\tau$ and $M$, as shown in Figure \ref{two cases of tau and M}.

	\begin{figure}[h]
		\centering
		\includegraphics[width=6cm]{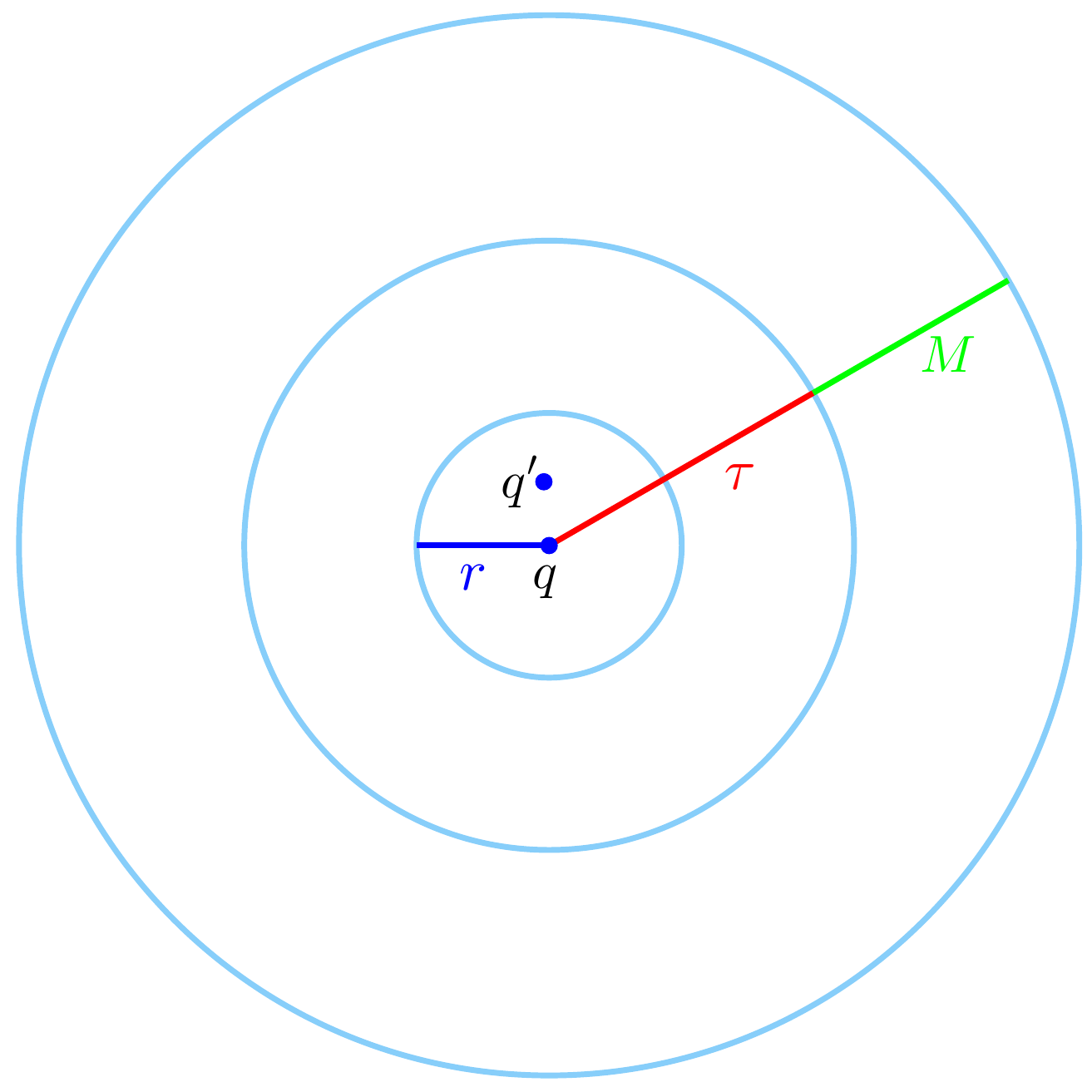}
		\hspace{1cm}
		\includegraphics[width=6cm]{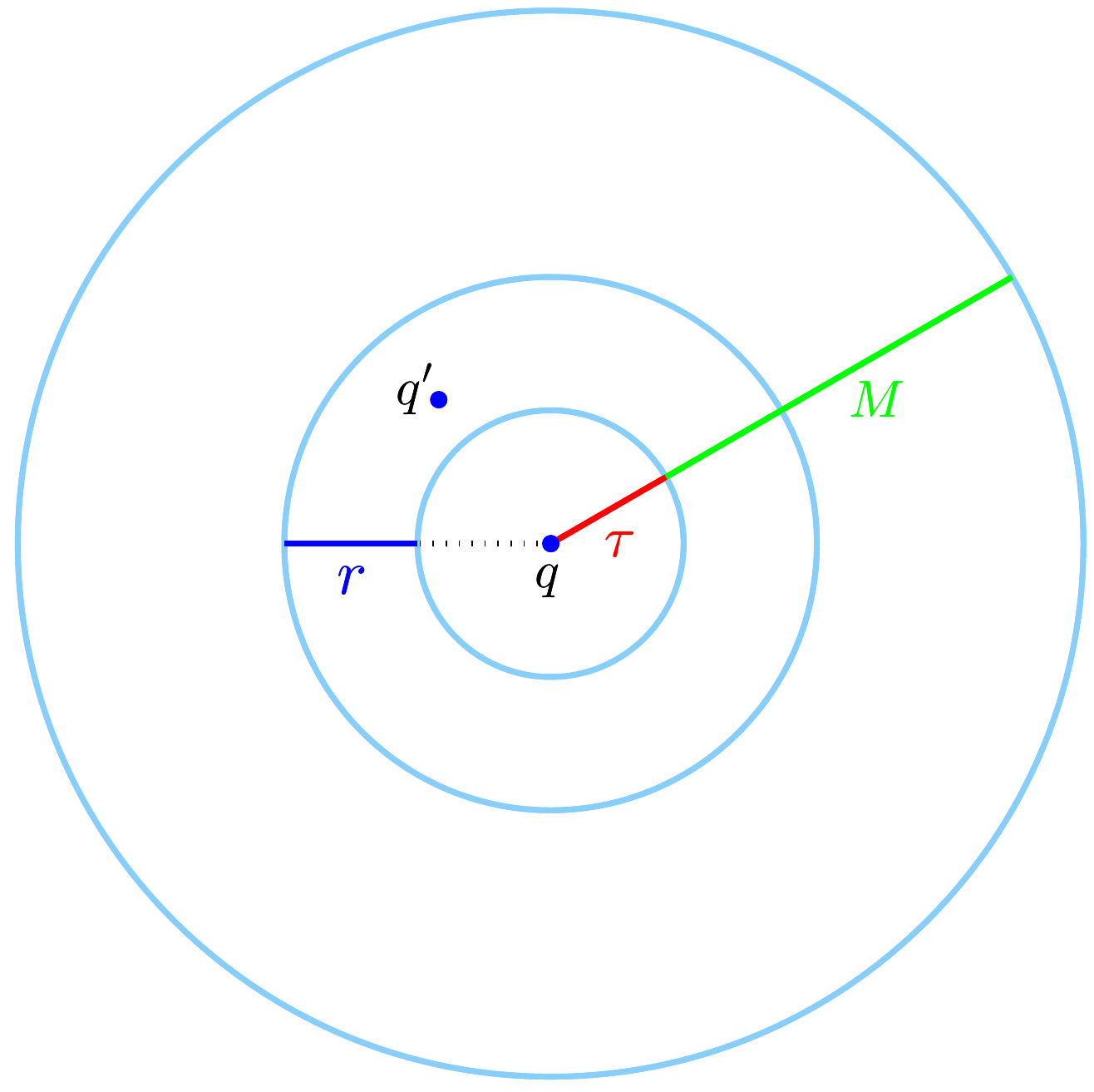}
		
		\caption{ Left: Case 1, $r=\frac{M}{8}\leq \tau$, and $q'\in B(q,r)$.
			Right: Case 2, $r=\frac{M}{8}> \tau$, and $q'\in B(q,\tau+r)$.}
		\label{two cases of tau and M}
	\end{figure}

\paragraph{Case 1: $\tau\geq \frac{M}{8}$.}
For any $q'\in B(q,\frac{M}{8}) := \{q'\in \R^d \mid \|q' - q_i \|\leq \frac{M}{8}\}$, we have
	$\tau+M=\dist(q,S_2) \leq \dist(q,q')+\dist(q',S_2)\leq \frac{M}{8}+\dist(q',S_2)$,
	which implies for all $q' \in B(q,\frac{M}{8})$
\[
	\dist(q',S_2)\geq \tau+M-\frac{M}{8}=\tau+\frac{7}{8}M.
\]
Similarly
$\dist(q', S_1)\leq \dist(q',q)+\dist(q,S_1)\leq \frac{M}{8}+\tau$ for all $\ q'\in B(q,\frac{M}{8})$.
Thus for all $q'\in B(q,\frac{M}{8})$
\[
	|\dist(q',S_2)-\dist(q',S_1)|\geq \dist(q',S_2)-\dist(q',S_1)\geq \tau+\frac{7}{8}M-(\tau+\frac{M}{8})
	=\frac{3}{4}M.
\]
	
\paragraph{Case 2: $0\leq \tau<\frac{M}{8}$.}	
For any $q'\in B(q,\tau+\frac{M}{8}):= \{q'\in \R^d \mid \dist(q',q)\leq \tau+\frac{M}{8} \}$,
	we have $\tau+M=\dist(q',S_2)\leq \dist(q,q')+\dist(q',S_2)\leq \tau+\frac{M}{8}+\dist(q',S_2)$,
	which implies for all $q' \in B(q,\tau+\frac{M}{8})$
\[
	\dist(q',S_2)\geq \frac{7}{8}M.
\]
Combined with $\tau<\frac{M}{8}$ and $\dist(q',S_1)\leq
	\dist(q',q)+\dist(q, S_1)\leq \tau+\frac{M}{8}+\tau=\frac{M}{8}+\frac{M}{8}+\frac{M}{8}\leq \frac{3}{8}M $ for all $q'\in B(q,\tau+\frac{M}{8})$, we have
\[	
	|\dist(q',S_2)-\dist(q',S_1)|\geq \dist(q',S_2)-\dist(q',S_1)
	\geq  \frac{7}{8}M-\frac{3}{8}M=\frac{M}{2}.
\]

Combining these two cases on $\tau$,  for all $q'\in B(q,\frac{M}{8})$
\[
|\dist(q',S_2)-\dist(q',S_1)|\geq\frac{M}{2}.
\]
Then since $B_\infty(q,\frac{r}{\sqrt{d}}) \subset B(q,r)$ for all $r\geq 0$, from
$C_q=\max\limits_{0< r \leq L} \frac{r^d}{L^d}\frac{n}{|Q\cap B_\infty(q,r)|}
	\geq (\frac{1}{8\sqrt{d}})^d \frac{M^d}{L^d}\frac{n}{|Q\cap B_\infty(q,\frac{M}{8\sqrt{d}})|}$, we can bound the denominator in $\sigma(q)$ as
\begin{align*}	
\frac{1}{n}\sum_{q' \in Q} f_{S_1,S_2}(q')
&\geq
\frac{1}{n}\sum_{q'\in Q\cap B_\infty(q,\frac{M}{8\sqrt{d}})} f_{S_1,S_2}(q')
=
\frac{1}{n}\sum_{q'\in Q\cap B_\infty(q,\frac{M}{8\sqrt{d}})} (\dist(q',S_1)-\dist(q',S_2))^2
\\ & \geq
\frac{1}{4}\frac{1}{n}  M^2 \Big|Q\cap B_\infty(q,\frac{M}{8\sqrt{d}})\Big|
\geq
\frac{1}{4}(\frac{1}{8\sqrt{d}})^d \frac{M^2}{C_q}\frac{M^d}{L^d}=\frac{1}{4}(\frac{1}{8\sqrt{d}})^d\frac{1}{C_q}\frac{M^{2+d}}{L^d},
\end{align*}	
which implies
\[	
\sigma(q)=\frac{M^2}{\frac{1}{n}\sum_{q'\in Q} f_{S_1,S_2}(q')}\leq 4 (8\sqrt{d})^d M^2 C_q\frac{L^d}{M^{2+d}}
	=4 (8\sqrt{d})^d C_q\frac{L^d}{M^d}.
\]

Combining this with $\sigma(q)\leq \frac{M^2}{\rho^2}$, we have
$\sigma(q)\leq \min \big(\frac{M^2}{\rho^2},4 (8\sqrt{d})^d  C_q \frac{L^d}{M^d}\big)$.
If $M^{2+d}\leq 4 (8\sqrt{d})^d  C_q \rho^2 L^d$, then
$\frac{M^2}{\rho^2}\leq 4 (8\sqrt{d})^d  C_q \frac{L^d}{M^d}$,
which means
$\sigma(q)
\leq
\min \big(\frac{M^2}{\rho^2}, 4 (8\sqrt{d})^d  C_q \frac{L^d}{M^d}\big)
=
\frac{M^2}{\rho^2}\leq  4^\frac{2}{2+d}(8\sqrt{d})^{\frac{2d}{2+d}} C_q^{\frac{2}{2+d}}\Big(\frac{L}{\rho}\Big)^{\frac{2d}{2+d}}$.
If
$M^{2+d}\geq 4 (8\sqrt{d})^d  C_q \rho^2 L^d$,
then
$4 (8\sqrt{d})^d  C_q \frac{L^d}{M^d} \leq \frac{M^2}{\rho^2}$, so we also have
$\sigma(q)
\leq
\min \big(\frac{M^2}{\rho^2},4 (8\sqrt{d})^d  C_q \frac{L^d}{M^d}\big)
=
4 (8\sqrt{d})^d  C_q \frac{L^d}{M^d}
\leq
4^\frac{2}{2+d}(8\sqrt{d})^{\frac{2d}{2+d}} C_q^{\frac{2}{2+d}}\Big(\frac{L}{\rho}\Big)^{\frac{2d}{2+d}}$.
\end{proof}

Hence, to bound the total sensitivity of $F(\Sh(L),\rho)$, we need a bound of $C_Q = \frac{1}{n}\sum_{q\in Q} C_q^{\frac{2}{2+d}}$.

\begin{lemma} \label{lemma: the bound of C_Q}
Suppose $Q\subset [0,L]^d$ of size $n$,  $\eta=\min_{q,q'\in Q,\ q\neq q'}\|q-q'\|_\infty$, and $C_Q$ is given by \eqref{def of C_1 and C_Q}. Then we have
\[
	C_Q\leq C_d \min\Big(\big(\log_2 \frac{L}{\eta}\big)^{\frac{2}{2+d}},\big(\frac{1}{d}\log_2 n \big)^{\frac{2}{2+d}}\Big),
\]
where $C_d=2^{d+1}$.
\end{lemma}
\begin{proof}
	
	We define $\widetilde{C}_Q:=\frac{1}{n}\sum_{q\in Q} C_q$, and using H\"{o}lder inequality we have
\[
C_Q
=
\frac{1}{n}\sum_{q\in Q}C_q^{\frac{2}{2+d}}
\leq
\frac{1}{n}\Big(\sum_{q\in Q} C_q\Big)^{\frac{2}{2+d}}n^{\frac{d}{2+d}}
=
\Big(\frac{1}{n}\sum_{q\in Q}C_q\Big)^{\frac{2}{2+d}}=(\widetilde{C}_Q)^{\frac{2}{2+d}}.
\]
So, we only need  to bound $\widetilde{C}_Q$.
	
	We define $r_q:=\arg \max_{0<r \leq L} \frac{r^d}{L^d}\frac{n}{|Q\cap B_\infty(q,r)|}$ for all $q\in Q$,
	$Q_i:=\{q\in Q \mid \frac{L}{2^{i+1}}<r_q\leq \frac{L}{2^i}\}$, and $A:=\{i\geq 0 \mid i \text{ is an integer and } |Q_i|>0\}$.
	
	For any fixed $i\in A$, we use $l_i:=\frac{L}{2^{i+1}}$ as the side length of grid cell to
	partition the region $[0,L]^d$ into $s_i=(\frac{L}{l_i})^d=2^{(i+1)d}$ grid cells: $\Omega_1.\cdots,\Omega_{s_i}$ where each $\Omega_j$ is a closed set, and define $Q_{i,j}:=Q_i\cap \Omega_j$. Then, $|Q_i\cap \bar{B}_\infty(q,l_i)|\geq |Q_{i,j}|$ for all $q\in Q_{i,j}$
	where $\bar{B}_\infty(q,l_i):=\{q'\in \R^d|\ \|q'-q\|_\infty\leq l_i\}$, and  we have
\begin{align*}	
\sum_{q\in Q_i}\frac{r_q^d}{L^d}\frac{1}{|Q_i\cap B_\infty(q,r_q)|}
&\leq
\sum_{q\in Q} \frac{L^d}{2^{id}L^d}\frac{1}{|Q_i\cap B_\infty(q,r_q)|}
\leq \frac{1}{2^{id}}\sum_{q\in Q_i}\frac{1}{ |Q_i\cap \bar{B}_\infty(q,l_i)|}
\\ &\leq
\frac{1}{2^{id}}\sum_{j\in[s_i],|Q_{i,j}|>0}\sum_{q\in Q_{i,j}}\frac{1}{|Q_i\cap \bar{B}_\infty(q,l_i)|}
\leq
\frac{1}{2^{id}}\sum_{j\in[s_i],|Q_{i,j}|>0}\sum_{q\in Q_{i,j}}\frac{1}{|Q_{i,j}|}
\\ &=
\frac{1}{2^{id}}\sum_{j\in[s_i],|Q_{i,j}|>0}\frac{|Q_{i,j}|}{|Q_{i,j}|}
\leq
\frac{s_i}{2^{id}}=\frac{2^{(i+1)d}}{2^{id}}=2^d.
\end{align*}
	
Then using the definitions of $\widetilde{C}_Q$ and $r_q$ we have
\begin{align*}	
\widetilde{C}_Q
&=
\sum_{q\in Q}\max_{0<r \leq L}\frac{r^d}{L^d}\frac{1}{|Q\cap B_\infty(q,r)|}
=
\sum_{q\in Q}\frac{r_q^d}{L^d}\frac{1}{|Q\cap B_\infty(q,r_q)|}
=
\sum_{i\in A}\sum_{q\in Q_i}\frac{r_q^d}{L^d}\frac{1}{|Q\cap B_\infty(q,r)|}
\\ &\leq
\sum_{i\in A}\sum_{q\in Q_i}\frac{r_q^d}{L^d}\frac{1}{|Q_i\cap B_\infty(q,r)|}
\leq
\sum_{i\in A}2^d=2^d |A|.
\end{align*}
	
We assert $r_q\geq L n^{-\frac{1}{d}}$ for all $q\in Q$. This is because for any $r\in (0,Ln^{-\frac{1}{d}})$ we have
\[
	\frac{r^d}{L^d}\frac{n}{|Q\cap B_\infty(q,r)|}\leq \frac{L^d}{n L^d}\frac{n}{1}=1\leq \frac{L^d}{L^d}\frac{n}{|Q\cap B_\infty(q,L)|},
\]
which implies the optimal $r_q\in [L n^{-\frac{1}{d}},L]$.
Moreover, since $r_q\geq \min_{q'\in Q,\ q'\neq q}\|q-q'\|_\infty\geq \eta$, we have
$r_q \geq \max(L n^{-\frac{1}{d}},\eta)$ for all $q\in Q$. If $i>\min\big(\log_2\frac{L}{\eta},\frac{1}{d}\log_2 n\big)$, then $\frac{L}{2^i}< \max(L n^{-\frac{1}{d}},\eta)\leq r_q$, and from the definition of $Q_i$ and $A$ we know $i\notin A$, which implies $|A|\leq 1+\min\big(\log_2\frac{L}{\eta},\frac{1}{d}\log_2 n\big)$.
Hence we obtain
$\widetilde{C}_Q\leq 2^{d+1} \min\big(\log_2\frac{L}{\eta},\frac{1}{d}\log_2 n\big)$
and using $C_Q = (\widetilde{C}_Q)^{\frac{2}{2+d}}$ we prove the lemma.
\end{proof}

Since $f_{S_1,S_2}\in F(\Sh (L),\rho)$, we know $f_{S_1,S_2}(q)\leq dL^2$ for all $q\in Q$ and
$\frac{1}{n}\sum_{q'\in Q} f_{S_1,S_2}(q')\geq \rho^2$, so $\sigma(q)\leq \frac{dL^2}{\rho^2}$ for all $q\in Q$.
Thus, we can expand $\frac{1}{|Q|} \sum_{q \in Q} \sigma(q)$ using Lemma \ref{sensitivity of one point} and factor out $C_Q$ using Lemma \ref{lemma: the bound of C_Q} to immediately obtain the following theorem about the total sensitivity of $F(\Sh(L),\rho)$.

\begin{theorem} \label{total sensitivity of F(S(L),rho)}
	Suppose $L>\rho>0$,  $Q=\{q_1,\cdots,q_n\}\subset [0,L]^d$ and
$\eta=\min_{q,q'\in Q,\ q\neq q'}\|q-q'\|_\infty$.
Then, we have
\[
	\SG(F(\Sh(L),\rho)) \leq \SG_Q
     = O\left(
\left(\frac{L}{\rho}\right)^{\frac{2d}{2+d}} \min \left( \log \frac{L}{\eta}, \log n, \left(\frac{L}{\rho}\right)^2 \right)^{\frac{2}{2+d}}
\right).
\]
\end{theorem}


From  Lemma \ref{sensitivity of one point} and Theorem \ref{total sensitivity of F(S(L),rho)},
using \cite{MLLJ2010}[Lemma 2.1]  we can obtain the following theorem.

\begin{theorem}\label{weak coreset of Q, geometric objects}
Let $L>\rho>0$,  $Q=\{q_1,\cdots,q_n\}\subset [0,L]^d$, $S_1,S_2\in \Sh (L)$ and $\dQ(S_1,S_2)\geq \rho$.
Suppose $\sigma(q)$ and $\SG_Q$ are defined in Lemma \ref{sensitivity of one point} and Theorem \ref{total sensitivity of F(S(L),rho)} respectively.
Then for $\delta,\eps \in (0,1)$ a $\sigma$-sensitive sampling of size
$N\geq \frac{\SG_Q}{\delta\eps^2}$
provides $\tilde Q$, a $(\rho,\eps,\delta)$-coreset; that is
with probability at least $1-\delta$, we have
\[
	(1-\eps) \dQ(S_1,S_2)\leq  \mathtt{d}_{\tilde{Q},W}(S_1,S_2)
	\leq(1+\eps) \dQ(S_1,S_2).
\]	
\end{theorem}

If $Q$ describes a continuous uniform distribution in $[0,L]^d$ (or sufficiently close to one, like points on a grid), then there exists an absolute constant $C>0$ such that $C_q \leq C$ for all $q\in Q$, then in Lemma \ref{sensitivity of one point}  $\sigma(q)\leq C_d \big(\frac{L}{\rho}\big)^{\frac{2d}{2+d}}$ for all $q\in Q$, and in Theorem \ref{total sensitivity of F(S(L),rho)} $\SG_Q \leq C_d \big(\frac{L}{\rho}\big)^{\frac{2d}{2+d}}$. So, for uniform distribution, the sample size of
 $Q$ in Theorem \ref{weak coreset of Q, geometric objects} is independent from the size of $Q$, and for $d=2$ the bound $\SG_Q = O(L/\rho)$ matches the lower bound in Lemma \ref{lem:TS-lb}.

\begin{corollary} \label{total sensitivity uniform ditribution}
If $Q$ describes the continuous uniform distribution over $[0,L]^d$, then the sample size in  Theorem \ref{weak coreset of Q, geometric objects} can be reduced to
$N=O\Big(\big(\frac{L}{\rho}\big)^{\frac{2d}{2+d}} \frac{1}{\delta \eps^2}\Big)$.

\end{corollary}

\paragraph*{Remark.}
To compute the upper bound of $\sigma(q)$ in Lemma \ref{sensitivity of one point}, we need to compute $C_q$ which can be obtained in
$O(n\log n)$ time. For any fixed $q\in Q$, we sort $Q\setminus \{q\}=\{q_1,\cdots,q_{n-1}\}$ according to their $l^\infty$ distance from $q$, so that $\|q-q_i\|_\infty\leq \|q-q_j\|_\infty$ for any $i<j$. Then for $i\in [n]$ we compute $\frac{r_i^d}{L^d}\frac{n}{i}$, where $r_i= \|q-q_i\|_\infty$ for $=i\in [n-1]$ and $r_n=L$, and choose the maximum value
of $\frac{r_i^d}{L^d}\frac{n}{i}$ as $C_q$.

\section{Strong Coresets for the Distance Between Trajectories}
\label{Approximate the Distance Between Trajectories}

In this section, we study the distance $\dQ$ defined on a subset of $\Sh(L)$: the collection of $k$-piecewise linear curves, and
use the framework in \cite{VDH2016} to construct a strong approximation for $Q$. We assume the multiset $Q$ contains $m$ distinct points $q_1,\cdots,q_m$, where each point $q_i$ appears $m_i$ times and $\sum_{i=1}^m m_i=n$. So, in this section $Q$ will be viewed as a
a set $\{q_1,\cdots,q_m \}$ (not a multiset) and each point $q\in Q$ has a weight $w(q_i)=\frac{m_i}{n}$.

Suppose $\T_k:=\{\gamma=\langle c_0,\cdots,c_k\rangle \mid c_i\in \R^d\}$ is the collection of all piecewise-linear curves with $k$ line segments in  $\R^d$. For $\gamma=\langle c_0,\cdots,c_k\rangle \in \T_k$, $\langle c_0,\cdots,c_k \rangle$ is the sequence of $k+1$ critical points of $\gamma$.
The value $\dist(q,\gamma) = \inf_{p \in \gamma} \|p-q\|$, and function $f_{\gamma_1,\gamma_2}(q) = (\dist(q,\gamma_1)-\dist(q,\gamma_2))^2$ are defined as before.  We now use weights $w(q_i)=\frac{m_i}{n}$ $\big(\sum_{q\in Q}w(q)=1\big)$ and the resulting distance is $\dQ(\gamma_1,\gamma_2) = \big(\sum_{q \in Q} w(q) f_{\gamma_1,\gamma_2}(q)\big)^{\frac{1}{2}}$.

For $L>\rho>0$, $Q=\{q_1,\cdots,q_m\}\subset \R^d$ , we define
\[
\c{X}_k^d(L,\rho)
:=
\left\{(\gamma_1,\gamma_2)\in \T_k\times\T_k \mid \gamma_1,\gamma_2\in \Sh(L),
\;
\dQ(\gamma_1,\gamma_2) \geq \rho\right\}.
\]

We next consider the sensitivity adjusted weights $w'(q) = \frac{\sigma(q)}{\SG_Q} w(q)$ and cost function
$g_{\gamma_1,\gamma_2}(q) = \frac{1}{\sigma(q)} \frac{f_{\gamma_1,\gamma_2}(q)}{\bar f_{\gamma_1, \gamma_2}}$.
These use the general bounds for sensitivity in Lemma \ref{sensitivity of one point} and Theorem \ref{total sensitivity of F(S(L),rho)}, with
as usual $\bar{f}_{\gamma_1,\gamma_2}=\sum_{q\in Q}w(q)f_{\gamma_1,\gamma_2}(q)$.
These induce an adjusted range space $(Q,\c{T}_{k,d}')$ where each element is defined
\[
T_{\gamma_1, \gamma_2, \eta} = \{q \in Q \mid w'(q) g_{\gamma_1, \gamma_2}(q) \leq \eta,  \; \gamma_1, \gamma_2 \in \c{X}_k^d(L,\rho)\}.
\]
Now to apply the strong coreset construction of Braverman \etal~\cite{VDH2016}[Theorem 5.5] we only need to bound the shattering dimension of $(Q,\c{T}_{k,d}')$.

Two recent results provide bounds on the VC-dimension of range spaces related to trajectories.  Given a range space $(X,\c{R})$ with VC-dimension $\nu$ and shattering dimension $\mathfrak{s}$, it is known that $\mathfrak{s} = O(\nu \log \nu)$ and $\nu = O(\mathfrak{s})$.  So up to logarithmic factors these terms are bounded by each other.
First Driemel \etal~\cite{DPP19} shows VC-dimension for a ground set of curves $\mathbb{X}_m$ of length $m$, with respect to metric balls around curves of length $k$, for various distance between curves.  The most relevant case is where $m=1$ (so the ground set are points like $Q$), and the Hausdorff distance is considered, where the VC-dimension in $d=2$ is bounded $O(k^2 \log (km)) = O(k^2 \log k)$ and is at least $\Omega(\max \{k, \log m\}) = \Omega(k)$.
Second, Matheny \etal~\cite{MXP19} considered ground sets $\mathbb{X}_k$ of trajectories of length $k$, and ranges defined by geometric shapes which may intersect those trajectories anywhere to include them in a subset.  The most relevant cases is when they consider disks, and show the VC-dimension is at most $O(d \log k)$, and have a proof that implies it is at least $\Omega(\log k)$; but this puts the complexity $k$ on the ground set not the query.
More specifically, neither of these cases directly imply the results for our intended range space, since ours involves a pair of trajectories.

\begin{lemma}\label{bound the demension for trajectory}
The shattering dimension of range space $(Q, \c{T}'_{k,d})$ is $O(k^3)$, for constant $d$.
\end{lemma}

\begin{proof}
Suppose $(\gamma_1,\gamma_2)\in \c{X}_k^d(L,\rho)$ and $\eta \geq 0$, where $\gamma_1=\langle c_{1,0},\cdots,c_{1,k}\rangle$
and $\gamma_2=\langle c_{2,0},\dots,c_{2,k}\rangle$,
then we can define the range $T_{\gamma_1,\gamma_2,\eta}$  as
\begin{align*}
T_{\gamma_1,\gamma_2,\eta}
 :=&
\{q\in Q \mid w'(q)g_{\gamma_1,\gamma_2}(q)\leq \eta\}
\\  =&
\{q\in Q \mid w(q)f_{\gamma_1,\gamma_2}(q)\leq \SG_Q \bar{f}_{\gamma_1,\gamma_2} \eta\}
\\ =& \nonumber
\{q\in Q \mid w(q)(\dist(q,\gamma_1)-\dist(q,\gamma_2))^2\leq \SG_Q \bar{f}_{\gamma_1,\gamma_2} \eta\}.
\end{align*}	
	
\begin{figure}[h]
	\centering
	\includegraphics[width=0.6\linewidth]{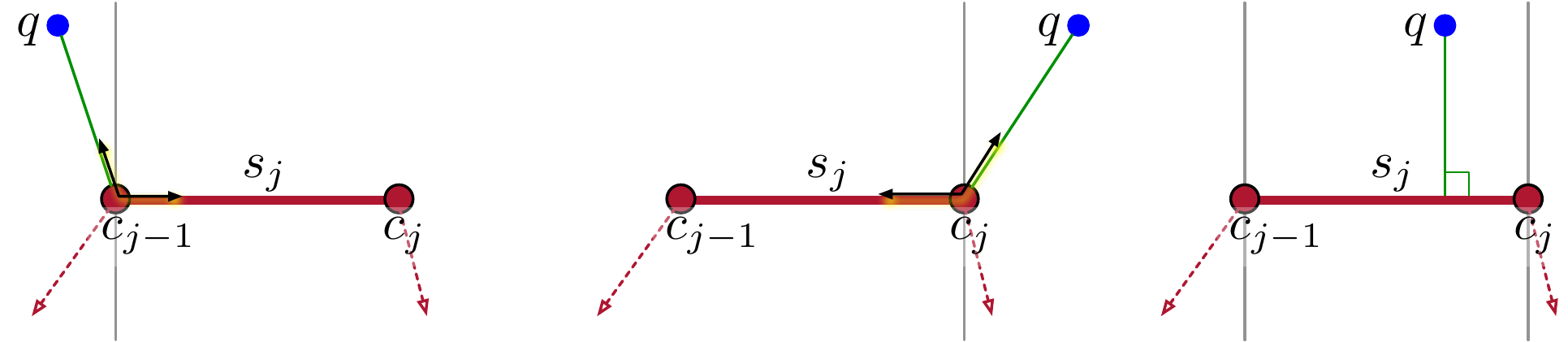}
		
	\caption{\label{fig:traj-VC}
	\small Illustration of the $\dist(q,s_j)$ from point $q$ to segment $s_j$.}
\end{figure}

For a trajectory $\gamma$ defined by critical points $c_0, c_1, \ldots, c_k$ for $j \in [k]$ define $s_j$ as the segment between $c_{j-1},c_j$ and $\ell_j$ as the line extension of that segment.
The distance between $q$ and a segment $s_{j}$ is illustrated in Figure \ref{fig:traj-VC} and defined
\[	
\xi_{j}:=\dist(q,s_{j})=\begin{cases}
	\dist(q,c_{j-1}), & \text{if } \langle c_{j}-c_{j-1},\; q-c_{j-1}\rangle\leq 0 \\
	\dist(q,c_{j}), & \text{if } \langle c_{j-1}-c_{j}, \;q-c_{j}\rangle\leq 0 \\
	\dist(q,\ell_{j}), & \text{otherwise} \\
	\end{cases}.
\]
Then $\dist(q,\gamma)=\min_{j\in [k]} \xi_{j}$.
For trajectories $\gamma_1$ and $\gamma_2$, specify these segment distances as $\xi_i^{(1)}$ and $\xi_i^{(2)}$, respectively.
Then the expression for $T_{\gamma_1,\gamma_2,\eta}$ can be rewritten as
\begin{align*}
T_{\gamma_1,\gamma_2,\eta}
&=
\{q\in Q \mid w'(q)g_{\gamma_1,\gamma_2}(q)\leq \eta\}
\\ &= \nonumber
\{q\in Q \mid w(q)(\min_{j\in [k]} \xi^{(1)}_{j}-\min_{j\in [k]} \xi^{(2)}_{j})^2\leq \SG_Q \bar{f}_{\gamma_1,\gamma_2} \eta\}
\\ &= \nonumber
\cup_{j_1,j_2\in[k]}\{q\in Q \mid \xi^{(1)}_{j_1}\leq \xi^{(1)}_{j}, \xi^{(2)}_{j_2}\leq \xi^{(2)}_{j} \text{ for all } j\in[k], \;
	w(q)(\xi^{(1)}_{i_1}-\xi^{(2)}_{j_2})^2\leq \SG_Q \bar{f}_{\gamma_1,\gamma_2} \eta\}
\\ & = \nonumber
\bigcup_{j_1,j_2\in[k]}
\left(\begin{array}{l}
    \big(\cap_{j\in[k],j\neq j_1} \{q\in Q \mid \xi^{(1)}_{j_1}\leq \xi^{(1)}_{j}\}\big)
    \\ \cap \;\;
    \big(\cap_{j\in[k],j\neq j_2} \{q\in Q|\ \xi^{(2)}_{j_2}\leq \xi^{(2)}_{j}\}\big)
    \\ \cap \;\;
    \{q\in Q \mid \sqrt{w(q)}(\xi^{(1)}_{j_1}-\xi^{(2)}_{j_2})\leq (\SG_Q \bar{f}_{\gamma_1,\gamma_2} \eta)^{\frac{1}{2}}\}
    \\ \cap \;\;
    \{q\in Q \mid \sqrt{w(q)}(\xi^{(2)}_{j_2}-\xi^{(1)}_{j_1})\leq (\SG_Q \bar{f}_{\gamma_1,\gamma_2} \eta)^{\frac{1}{2}}\}
\end{array}\right).
\end{align*}
This means set $T_{\gamma_1,\gamma_2,\eta}$ can be decomposed as the union and intersection of $O(k^3)$ simply-defined subsets of $Q$.  Specifically looking at the last line, this can be seen as the union over $O(k^2)$ sets (the outer union), and the first two lines are the intersection of $O(k)$ sets, and the last two lines inside the union are the intersection with one set each.

Next we argue that each of these $O(k^3)$ simply defined subsets of $Q$ can be characterized as an element of a range space.  By standard combinatorics~\cite{s-gaa-11,AB99} (and spelled out in Lemma \ref{bound the dimension of range space intersection and union}), the bound of the shattering dimension of the entire range space is $O(k^3)$ times the shattering dimension of any of these simple ranges spaces.


To get this simple range space shattering dimension bound, we can use a similar linearization method as presented in the proof of Lemma \ref{bound the demension for hyperplane}.
For any simple range space $\c{R}$ determined by the set decomposition of $T_{\gamma_1,\gamma_2,\eta}$,
we can introduce new variables $c_0\in \R, z,c\in \R^{d'}$, where $z$  depends only on $q$, and $c_0$, $c_i$ depend only on
    $\gamma_1,\gamma_2$ and $r$, and $d'$ only depends on $d$.	 Here, $Q$ is a fixed set and thus $\SG_Q$ is a constant.
	By introducing new variables we can construct an injective map $\varphi: Q\mapsto \R^{d'}$, s.t. $\varphi(q)=z$. There is also an injective
    map	from $\c{R}$ to $\{\{z\in \varphi(Q) \mid c_0+z^Tc\leq 0\} \mid c_0\in \R, c\in \R^{d'}\}$. Since the shattering dimension of the range space $(\R^{d'},\c{H}^{d'})$, where $\c{H}^{d'}=\{h \text{ is a halfspace in } \R^{d'}\}$, is $O(d')$, we have the shattering dimension  of$(Q,\c{R})$ is $O(d')\leq C_d$ where $C_d$ is a positive constant depending only on $d$.  Piecing this all together we obtain $C_d k^3$ bound for the shattering dimension of $(Q,\c{T}'_{k,d})$.
\end{proof}

Now, we can directly apply Lemma \ref{bound the demension for trajectory} and \cite{VDH2016}[Theorem 5.5] to get a $(\rho,\eps,\delta)$-strong coreset for $\c{X}_k^d(L,\rho)$.

\begin{theorem}\label{strong coreset of Q, trajectories}
Let $L>\rho>0$,  $Q \subset [0,L]^d$, and consider trajectory pairs $\c{X}_k^d(L,\rho)$.
Suppose $\sigma(q)$ and $\SG_Q$ are defined in Lemma \ref{sensitivity of one point} and Theorem \ref{total sensitivity of F(S(L),rho)} respectively.
Then for $\delta,\eps \in (0,1)$ a $\sigma$-sensitive sampling of size
$N = O(\frac{\SG_Q}{\eps^2}(k^3 \log \SG + \log \frac{1}{\delta}))$
provides $\tilde Q$, a strong $(\rho,\eps,\delta)$-coreset; that is
with probability at least $1-\delta$, for all pairs $\gamma_1, \gamma_2 \in \c{X}_k^d(L,\rho)$ we have
\[
	(1-\eps) \dQ(\gamma_1,\gamma_2)\leq  \mathtt{d}_{\tilde{Q},W}(\gamma_1,\gamma_2)
	\leq(1+\eps) \dQ(\gamma_1,\gamma_2).
\]	
\end{theorem}

\section{Trajectory Reconstruction}
\label{sec:recon}

In Section \ref{Approximate the Distance Between Trajectories}, we use $Q$ to convert a piecewise-linear curve $\gamma$
to a vector $v_Q(\gamma)$ in $\R^{|Q|}$, and in this section we study how to recover $\gamma$
from $Q$ and $v_Q(\gamma)$, and we only consider $\gamma$ in $\R^2$.

Let $\T:=\{\gamma=\langle c_0,\cdots,c_k\rangle \mid c_i\in \R^2, k\geq 1\}$  be the set of all piecewise-linear curves in $\mathbb{R}^2$. Each curve in $\T$ is specified by a series of critical points $\langle c_0, c_1, \ldots, c_k \rangle$, and  $k$ line segments $s_1, s_2, \ldots, s_k$,  where $s_i$ is the line segment $\overline{c_{i-1}c_i}$.

For a curve $\gamma \in \T$ and $\tau>0$
we define a family of curves $\T_{\tau} \subset \T$ s.t. each $\gamma \in \T_{\tau}$ has two restrictions:
\begin{itemize}\denselist
\item[(R1)]
Each angle $\angle_{[c_{i-1}, c_i, c_{i+1}]}$ about an internal critical point $c_i$ is non-zero (i.e., in $(0, \pi)$).
\item[(R2)]
Each critical point $c_i$ is \emph{$\tau$-separated}, that is the disk $B(c_i, \tau) = \{x \in \R^2 \mid \|x-c_i\| \leq \tau\}$ only intersects the two adjacent segments $s_{i-1}$ and $s_i$ of $\gamma$, or one adjacent segment for endpoints (i.e., only the $s_1$ for $c_0$ and $s_k$ for $c_k$).
\end{itemize}

We next restrict that all curves (and $Q$) lie in region $\Omega\subset \R^2$.
Let $\T_{\tau}(\Omega)$ be the subset of $\T_{\tau}$ where all curves $\gamma$
have all critical points within $\Omega$, and in particular,
no $c_i \in \gamma$ is within a distance $\tau$ of the boundary of $\Omega$.
Now for $i > 0$, define an infinite grid $G_i = \{g_\eta \in \R^2 \mid  g_\eta = \eta v \text{ for } v = (v_1,v_2) \in \mathbb{Z}^2 \}$, where $\mathbb{Z}$ is all integers.

Suppose $\eta\leq \frac{\tau}{32}$,  $Q = G_\eta \cap \Omega=\{q_1,\cdots,q_n\}$, $\gamma\in \T_\tau(\Omega)$,
$v_i = \min_{p \in \gamma} \|q_i - p\|$ and $v_Q(\gamma) =(v_1,\dots,v_n)$.
We define some notations that are used in this section for the
 implied circle $C_i:=\{x\in \R^2|\  \|x-q_i\|=v_i\}$,
 the closed disk $B_i:=\{x\in \R^2|\  \|x-q_i\|\leq v_i\}$, and
 the open disk $\dot{B}_i:=\{x\in \R^2|\  \|x-q_i\|< v_i\}$ around each $q_i$ or radius $v_i$.
 When the radius is specified as $r$ (with perhaps $r \neq v_i$), then we, as follows, denote the associated circle $C_{i,r}$, closed disk $B_{i,r}$, and open disk $\dot{B}_{i,r}$ around $q_i$.

For $Q$, $\gamma\in \T_\tau(\Omega)$ and $v_Q(\gamma)$ we have the following three observations.
\begin{itemize}\denselist
\item[(O1)]
In any disk with radius less than $\tau$, there is at most one critical point of $\gamma$; by (R2).
\item[(O2)]
If a point moves along $\gamma$, then it can only stop or change direction at critical points of $\gamma$.
\item[(O3)]
For any $q_i\in Q$, $\gamma$ cannot go into  $\dot{B}_i$. Moreover, $C_i$ must contain at least one point of $\gamma$,
and if this point is not a critical point, then $\gamma$ must be tangent to $C_i$ at this point.
\end{itemize}

The restriction (R2) only implies if there is a critical point of $\gamma$, then in its neighborhood $\gamma$ has at most two line segments. However, if there is no critical point in a region, then the shape of $\gamma$ can be very complicated in this region, so
we need to first identify the regions that contain a critical point.

The entire algorithm is overviewed in Algorithm \ref{alg:full-recover}.
For each critical point $c\in \gamma$, there exists $q\in Q$ such that $\dist(q,c)<\eta $.
So to recover $\gamma$, we first traverse $\{q_i\in Q \mid v_i<\eta\}$ and use \textsc{isCritical}$(q_i)$ (Algorithm \ref{the existence of critical point}) to solve the decision problem of if there is a critical point in $B_{i,3\eta}$.
Whenever there is a critical point in $B_{i,3\eta}$, we then use \textsc{FindCritical}$(q_i)$ (Algorithm \ref{find a criticall point}) to find it -- collectively, this finds all critical points of $\gamma$.
Finally, we use \textsc{DetermineOrder} (Algorithm \ref{alg:DOC}) to determine the order of all critical points of $\gamma$, which recovers $\gamma$.

\begin{algorithm}
	\caption{Recover $\gamma\in \T_{\tau}(L)$ from $Q$ and $v_Q(\gamma)$}
	\label{alg:full-recover}
	\begin{algorithmic}
		\STATE Initialize $Q_\eta:=\{q_i\in Q \mid v_i<\eta\}$, close set $Q_{\text{r}}:=\emptyset$, endpoints $E=\emptyset$ and critical points $A:=\emptyset$.
		\FOR {each $q_i\in Q_\eta$}
		\IF { $q_i\in Q_{\text{r}}$ or \textsc{isCritical}$(q_i)$=\textsc{False}}
		\STATE \textbf{continue}
		\ENDIF
		\STATE Let $(c,S):=\textsc{FindCritical}(q_i)$.
		\IF {$|S|=1$}
		\STATE  $E:=E\cup \{(c,S)\}$.    \ \ \ \ \ \ // $c$ is an endpoint of $\gamma$
		\ENDIF
		\STATE Let $A:=A\cup \{(c,S)\}$ and $Q_{\text{r}}:=Q_{\text{r}}\cup \big(Q_\eta \cap B_{c,16\eta}\big)$.   \ \ \ \ \ \ // aggregate critical points
		\ENDFOR		
		\RETURN $\gamma:=\textsc{DetermineOrder}(E,A)$
	\end{algorithmic}
\end{algorithm}

\paragraph{Existence of critical points.}

In Algorithm \ref{the existence of critical point}, we consider the common tangent line of $C_i$ and $C_j$ for all $q_j$ in a neighborhood of $q_i$. If no common tangent line can go through $B_{i,3\eta}$ without going into the interior of any other circle centered in $B_{i,3\eta}$, then it implies there is a critical point of $\gamma$ in $ B_{i,3\eta}$.

\begin{algorithm}
	\caption{\textsc{isCritical}$(q_i)$:  Determine the existence of critical point in $B_{i,3\eta}$}
	\label{the existence of critical point}
	\begin{algorithmic}
		\FOR {each $q_j\in Q_{i,3\eta}\setminus \{q_i\}$}
		  \STATE Let $\ell_{i,j}$ be a common tangent line of $C_i$ and $C_j$.
		  \IF {$\ell_{i,j}$ does not intersect with $\dot{B}_k$ for all $q_k \in Q_{i,3\eta} \setminus \{q_i, q_j\}$}
		    \RETURN \textsc{False}
          \ENDIF
		\ENDFOR
		\RETURN \textsc{True}  \ \ \ \ \ \  // there must be a critical point in $B_{i,3\eta}$
	\end{algorithmic}
\end{algorithm}

\begin{figure}[h]
	\centering
	\includegraphics[width=5.2cm]{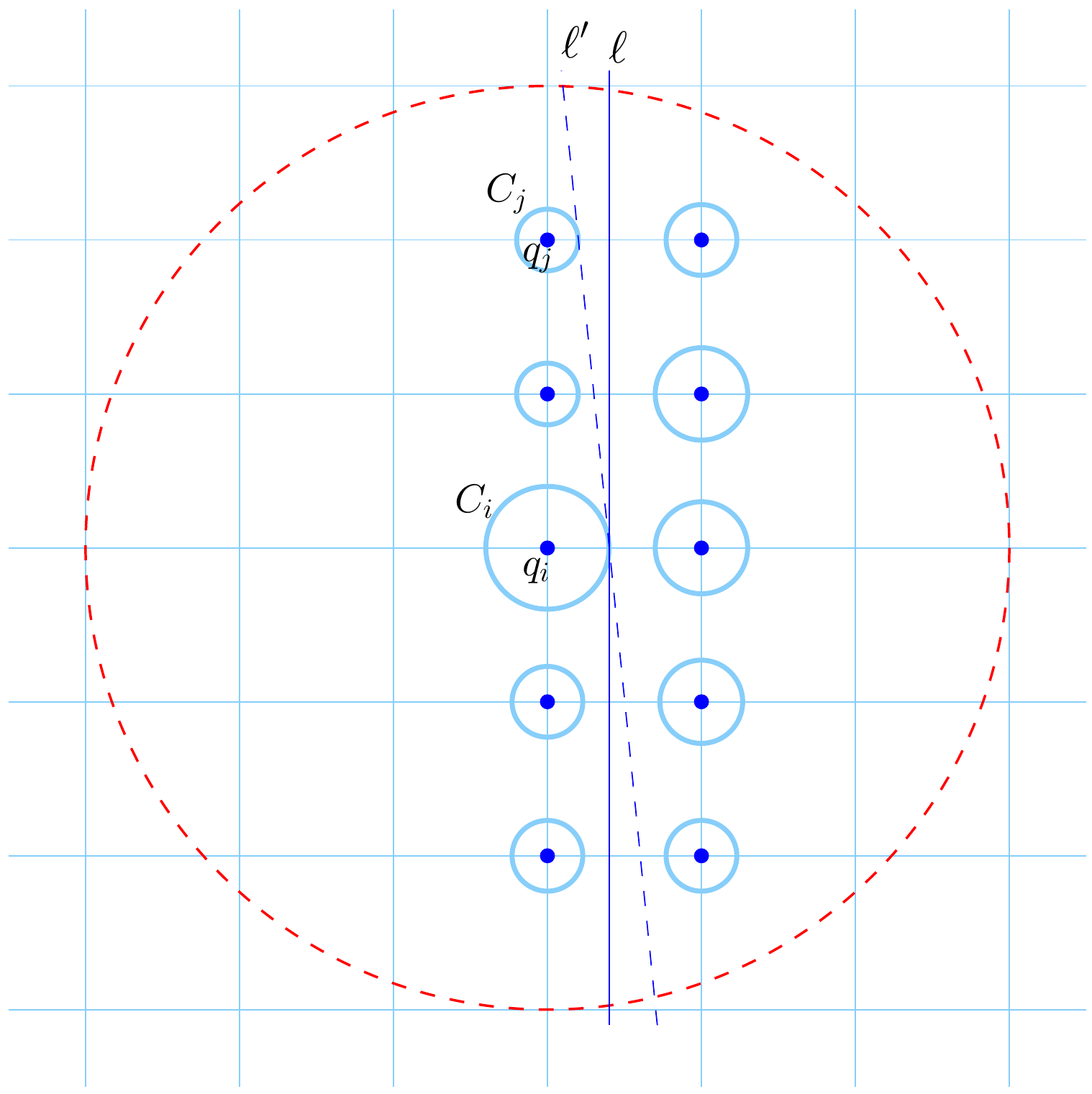}
	\includegraphics[width=5.2cm]{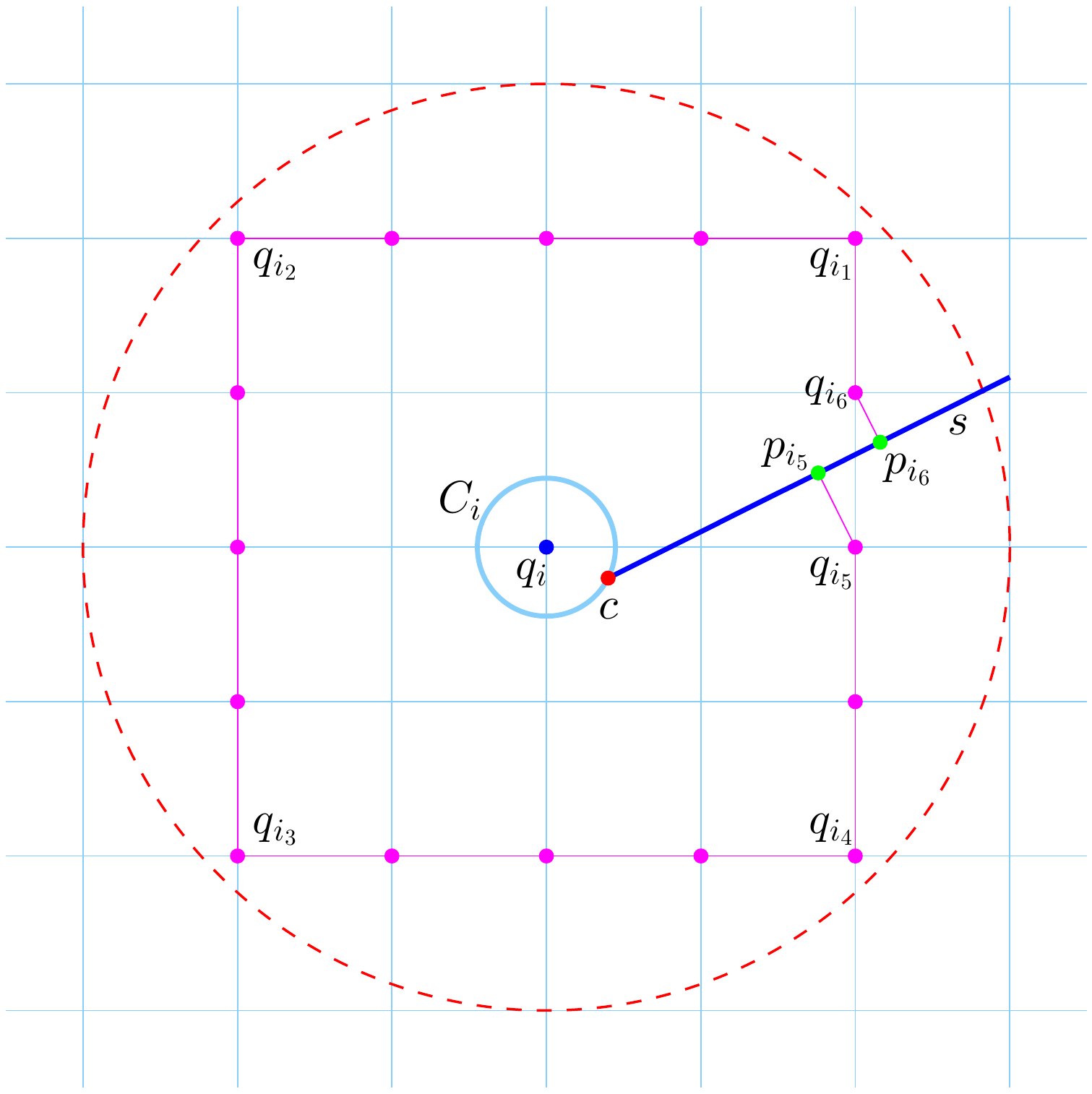}
	\includegraphics[width=5.2cm]{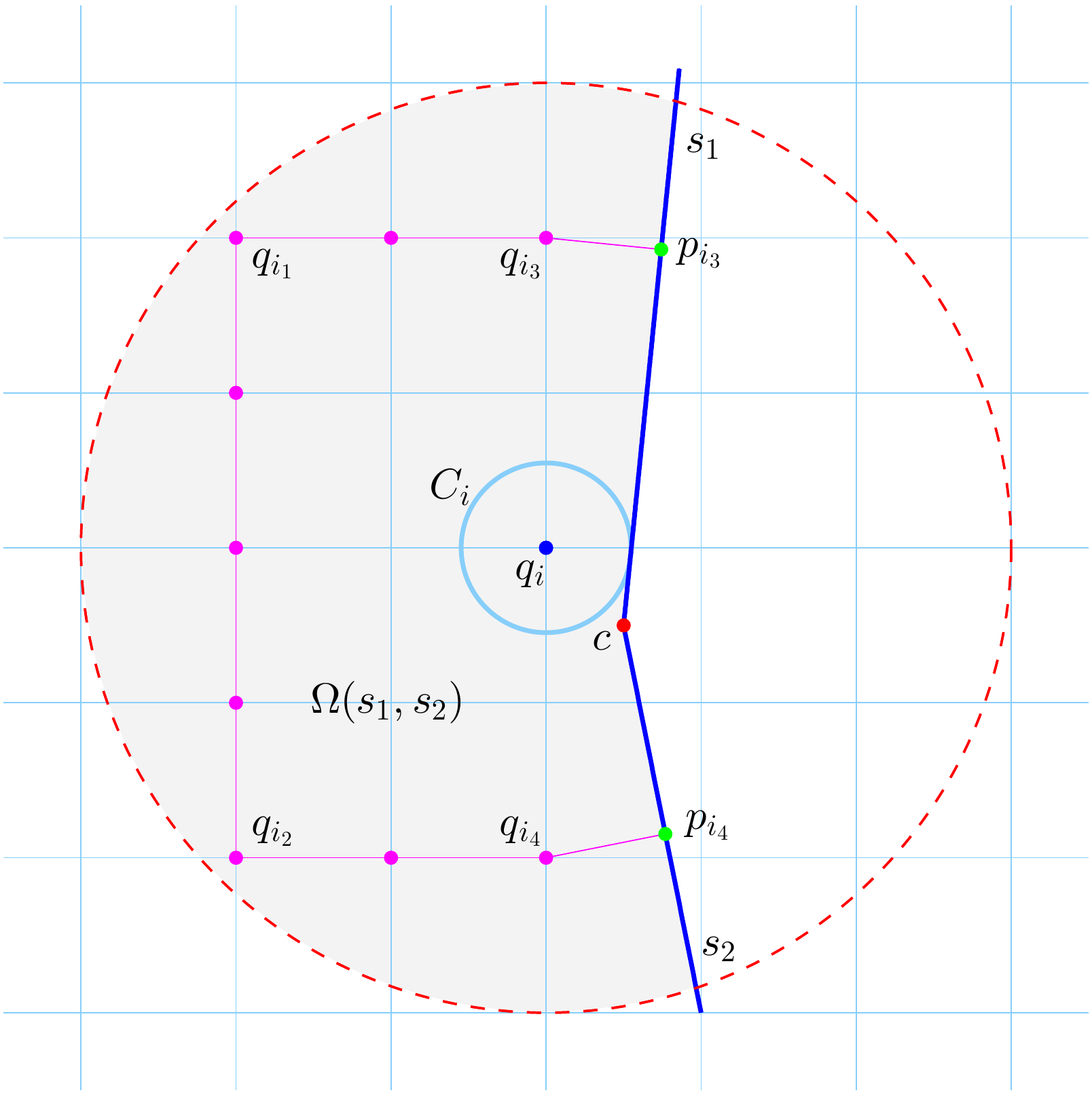}
	
	\vspace{-0.1in}
	\caption{Left: $l$ is tangent to $C_i$. Rotate $l$ around $C_i$ until it is tangent to some $C_j$.
	         Center: $c$ is an endpoint of $\gamma$. Right: $c$ is an internal critical point of $\gamma$.
	         In center and right figures, no tangent line of $C_i$ can go through  $B_{i,3\eta}$
	         without intersecting with the pink curve.}
	\label{fig critical point existence}
	
\end{figure}

\begin{lemma}  \label{lemma critical point existence}
	Suppose $q_i\in Q$ and $v_i<\eta$.
	If \textsc{isCritical}$(q_i)$ (Algorithm \ref{the existence of critical point}) returns \textsc{True}, then there must be a critical point of $\gamma$ in $B_{i,3\eta}$.
	Moreover, for any critical point $c\in \gamma$ there exists some $q_i\in Q$ such that $v_i<\eta$ and \textsc{isCritical}$(q_i)$ (Algorithm \ref{the existence of critical point}) returns \textsc{True} for the input $q_i$.
\end{lemma}

\begin{proof}
	
If Algorithm \ref{the existence of critical point} returns \textsc{True}, then no common tangent of $C_i$ and $C_j$ ($q_j\in Q_{i,3\eta}$) can go through $B_{i,3\eta}$ without intersecting with some $\dot{B}_k$ for $q_k \in Q_{i,3\eta}$.
This implies no tangent line of $C_i$ can go through $B_{i,3\eta}$ without intersecting with some $\dot{B}_k$ for $q_k \in Q_{i,3\eta}$. Otherwise, as shown in Figure \ref{fig critical point existence}(Left), suppose tangent line $\ell$ can go through $B_{i,3\eta}$, then we can rotate $\ell$ around $C_i$ to line $\ell'$ s.t. $\ell'$ is tangent to some $C_j$ ($q_j\in Q_{i,3\eta}$) but does not intersect with  any $\dot{B}_k$ ($q_k \in Q_{i,3\eta}$), which leads to a contradiction.
So, if there is no critical point on $C_i$ then (O3) implies one line segment of $\gamma$ must be tangent
 to $C_i$, but Algorithm \ref{the existence of critical point} checks that no tangent line of $C_i$ can go through $B_{i,3\eta}$ and thus from (O2) we know $\gamma$ must have a critical point in $B_{i,3\eta}$.
	
If $c\in \gamma$ is a critical point, then there are two possibilities: $c$ is an endpoint of $\gamma$, or $c$ is an internal critical point of $\gamma$.
	
If $c$ is an endpoint, let $q_i=(x_i,y_i)$ be the closest point in $Q$ to $c$.
Obviously we have $v_i<\eta$, and there is only one line segment $s$ of $\gamma$ in $B_{i,3\eta}$.
We consider the points set  $S_{i,2\eta}:=\{(x_i+k_1\eta,y_i+k_2\eta) \mid \|(k_1,k_2)\|_\infty = 2 \}$,
i.e. the pink points in Figure \ref{fig critical point existence}(Center).
Without loss of generality, we assume $q_{i_5}=(x_i+2\eta,y_i)$ and $q_{i_6}=(x_i+2\eta,y_i+\eta)$ are the two closest points in $S_{i,2\eta}$ to $s$, and their projection on $s$ are $p_{i_5}$ and $p_{i_6}$ respectively (two green points in Figure \ref{fig critical point existence}(Center)).
Let $q_{i_1}=(x_i+2\eta,y_i+2\eta)$, $q_{i_2}=(x_i-2\eta,y_i+2\eta)$, $q_{i_3}=(x_i-2\eta,y_i-2\eta)$ and $q_{i_4}=(x_i+2\eta,y_i-2\eta)$ be the four pink corners.
Since the radius of $C_i$ is $v_i<\eta$, we know any tangent line of $C_i$ must intersect with the piecewise-linear curve
	      $\langle p_{i_6},q_{i_6},q_{i_1},q_{i_2},q_{i_3},q_{i_4},q_{i_5},p_{i_5} \rangle$
before it passes completely through $B_{i,3\eta}$.
However, the curve
$\langle p_{i_6},q_{i_6},q_{i_1},q_{i_2},q_{i_3},q_{i_4},q_{i_5},p_{i_5} \rangle$
is covered (except points $p_{i_6}$ and $p_{i_5}$) by open disks $\dot{B}_{k}$ whose centers are in $q_k \in S_{i,2\eta} \subset Q_{i,3\eta}$.
So, no tangent line of $C_i$ can go through $B_{i,3\eta}$ without intersecting with some $\dot{B}_k$ for $q_k \in Q_{i,3\eta}$.
	
If $c$ is an internal critical point, then there are two line segments $s_1,s_2$ in $B_{i,3\eta}$.
From (R1) we know the angle between $s_1$ and $s_2$ is less than $\pi$, and we define
	      $\Omega(s_1,s_2):=\{p\in B_{i,3\eta} \mid p \text{ is outside the interior angle}$ $\text{region formed by } s_1 \text{ and } s_2\}$.
Let $q_i=(x_i,y_i)$ be the closest point in $\Omega(s_1,s_2)$ to $c$, and $S_{i,2\eta}$ be defined in the same way as before. We have $v_i<\eta$.
We consider the points set $S_{i,2\eta} \cap \Omega(s_1,s_2)$, i.e. those pink points in Figure \ref{fig critical point existence}(Right).
Without loss of generality, we assume $q_{i_3}=(x_i,y_i+2\eta)$ and $q_{i_4}=(x_i,y_i-2\eta)$ are two closest points in $S_{i,2\eta} \cap \Omega(s_1,s_2)$ to $s_1$ and $s_2$ respectively, and their projection on $s_1$ and $s_2$ are $p_{i_3}$ and $p_{i_4}$ respectively (two green points in Figure \ref{fig critical point existence}(Right)).
In this setting, let $q_{i_1}=(x_i-2\eta,y_i+2\eta)$ and $q_{i_2}=(x_i-2\eta,y_i-2\eta)$ be the corner points of $S_{i,2\eta}$.
Since the radius of $C_i$ is $v_i<\eta$ and the angle formed by $s_1$ and $s_2$ is less than $\pi$, we know any tangent line of $C_i$ must intersect with the piecewise-linear curve
  $\langle p_{i_4},q_{i_4},q_{i_2},q_{i_1},q_{i_3},p_{i_3} \rangle$
before go through $B_{i,3\eta}$.
However, the curve
  $\langle p_{i_4},q_{i_4},q_{i_2},q_{i_1},q_{i_3},p_{i_3} \rangle$
is covered by open disks $\dot{B}_k$ whose centers are $q_k \in S_{i,2\eta} \cap \Omega(s_1,s_2) \subset Q_{i,3\eta}$.
So, we know no tangent line of $C_i$ can pass entirely through $B_{i,3\eta}$ without intersecting with some $\dot{B}_k$ for $q_k \in Q_{i,3\eta}$.

Thus, if $c$ is a critical point of $\gamma$, Algorithm \ref{the existence of critical point} will return \textsc{True} for some $q_i\in Q$ with $v_i<\eta$.
\end{proof}

\paragraph{Finding a critical point.}

If there is a critical point $c$ in $B_{i,3\eta}$, then using (R2) we know in the neighborhood of $c$, $\gamma$ has a particular pattern: it either has one line segment, or two line segments.
We will need two straightforward subfunctions:
\begin{itemize}
  \item \textsc{FCT} (\emph{Find Common Tangents}) takes in three grid points $q_i,q_j,q_k$, and returns the all common tangent lines of $C_j$ and $C_k$ which do not intersect the interior of disks $\dot B_l$ of an disk associated with a point $q_l \in Q_{i,8\eta}$.    This generates a feasible superset of possible nearby line segments which may be part of $\gamma$.
  \item \textsc{MOS} (\emph{Merge-Overlapping-Segments}) takes a set of line segments, and returns a smaller set, merging overlapping segments.  This combines the just generated potential line segments of $\gamma$.
\end{itemize}
Now in Algorithm \ref{find a criticall point}, for each pair $q_j,q_k\in B_{i,8\eta}$, we first use \textsc{FCS} to find the common tangent line of $C_j, C_k$ that could be segments of $\gamma$, and then use \textsc{MOS} to reduce this set down to a minimal set of possibilities $S_m$.
By definition, there must be a critical point $c$, and thus can be at most $2$ actual segments of $\gamma$ within $B_{i,8\eta}$, so we can then refine $S_m$.
We first check if $c$ is an endpoint, in which case there must be only one valid segment.  If not, then there must be $2$, and we need to consider all pairs in $S_m$.  This check can be done by verifying that \emph{every} $C_k$ for $q_k \in Q_{i,8\eta}$ is tangent to the associated ray $\mathsf{ray}(s)$ (for an endpoint) or for the associated rays $\mathsf{ray}(s)$ and $\mathsf{ray}(s')$ for their associated segment pairs (for an internal critical point).

\begin{algorithm}
	\caption{\textsc{FindCritical}$(q_i)$:  Find a critical point in $B_{i,3\eta}$}
	\label{find a criticall point}
	\begin{algorithmic}
		\STATE Let  $Q_{i,8\eta}:=Q\cap B_{i,8\eta}$ and $S_{\text{t}}:=\emptyset$.
		\FOR {each pair $q_j,q_k\in Q_{i,8\eta}$}
		\STATE $S_{\text{t}}:=S_{\text{t}}\cup \textsc{FCT}(q_i,q_j,q_k)$
		\ENDFOR
		\STATE $S_{\text{m}}:=\textsc{MOS}(S_\text{t})$.
		\FOR{each $s \in S_{\text{m}}$}
		\STATE Extend $s$ to ray $\mathsf{ray}(s)$ with endpoint $c$ where it first enters $\dot{B}_{k}$ for some $q_{k}\in Q_{i,8\eta}$.
		\IF  {for all $q_j \in Q_{i,8\eta}$ either $c\in C_j$ or $C_j$ is tangent to $\mathsf{ray}(s)$ (\textsc{EndPoint})}
		\RETURN ($c$, $\{s\}$) \ \ \ \ \ \  // $c$ is an endpoint of $\gamma$
		\ENDIF
		\ENDFOR
		\FOR{each pair $s, s' \in S_{\text{m}}$}
		\STATE  Extend to lines $\ell(s), \ell(s')$.
		\STATE \textbf{if} $\ell(s)$ and $\ell(s')$ do not intersect in $B_{i,8\eta}$ \textbf{continue}
		\STATE  Set $c = \ell(s) \cap \ell(s')$, and define rays from $c$ containing $s$ and $s'$ as $\mathsf{ray}(s)$ and $\mathsf{ray}(s')$.
		\IF  {for all $q_k \in Q_{i,8\eta}$ either $c\in C_k$ or $C_k$ is tangent to $\mathsf{ray}(s)$ or $\mathsf{ray}(s')$	(\textsc{InternalPoint})}
		\RETURN ($c$,  $\{s,s'\}$) \ \ \ \ \ \ // $c$ is an internal critical point of $\gamma$
		\ENDIF
		\ENDFOR
	\end{algorithmic}
\end{algorithm}

%
%
%
%
%

\begin{figure}[h]
	\centering
	\includegraphics[width=5.2cm]{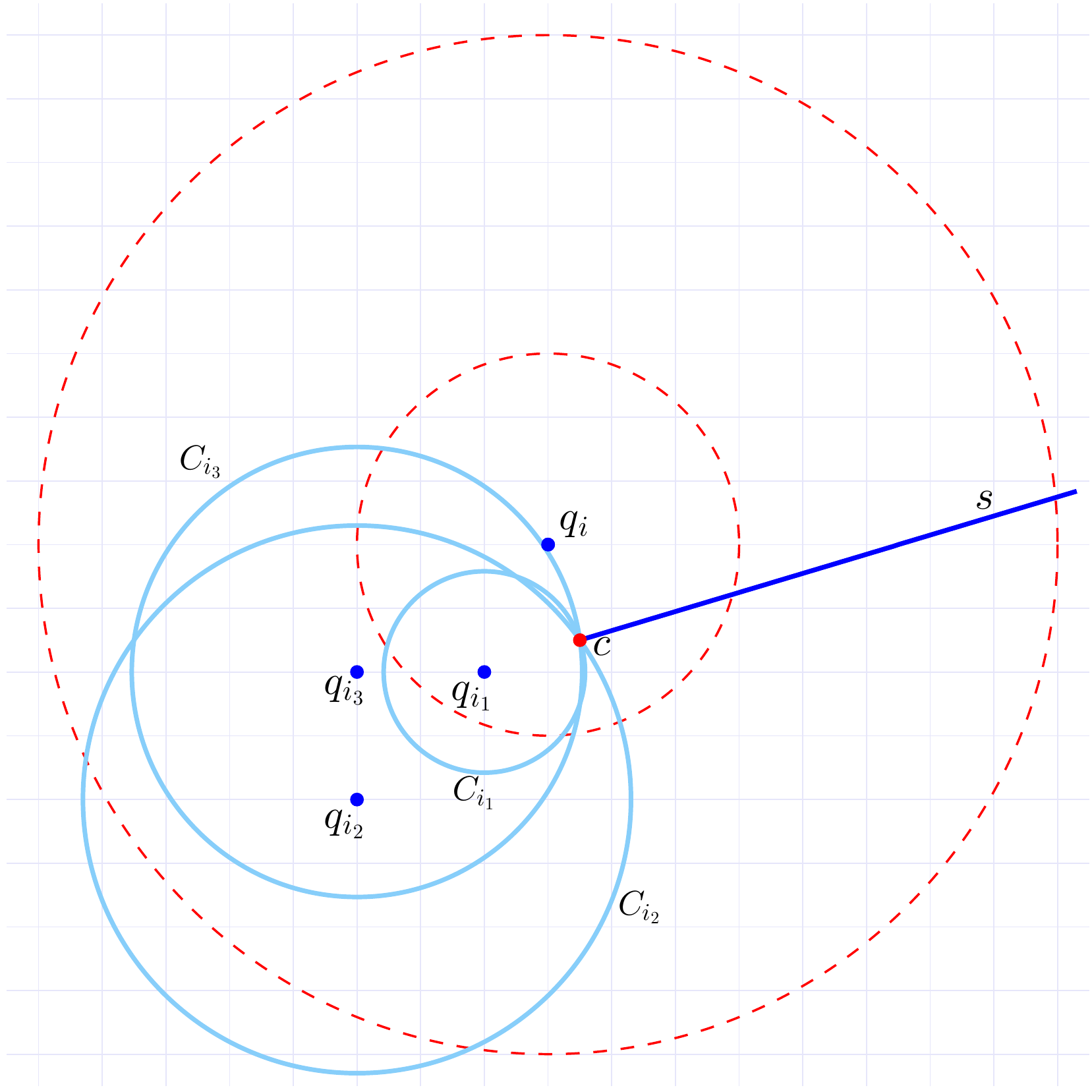}
	\hspace{0.1cm}
	\includegraphics[width=5.2cm]{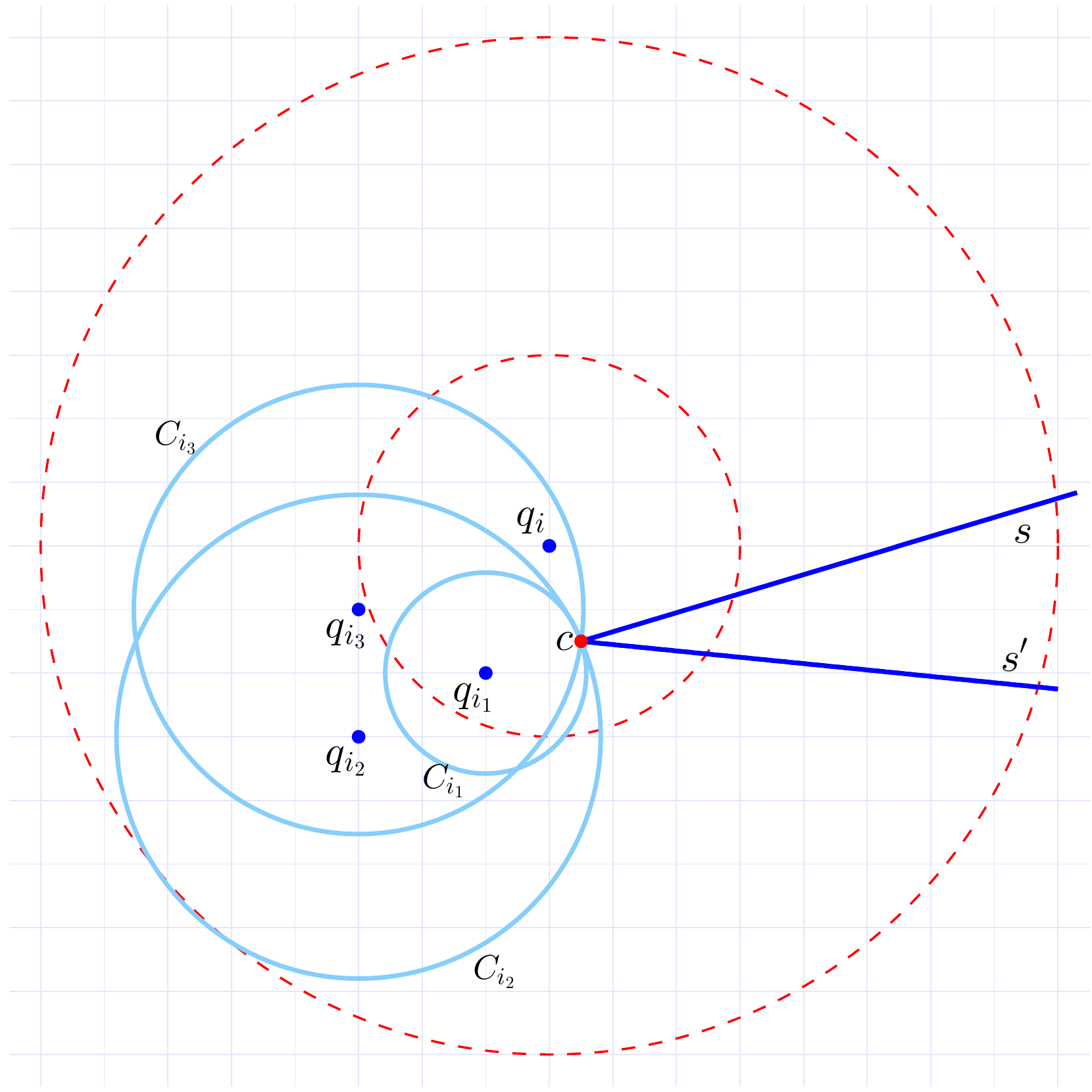}
	\hspace{0.1cm}
	\includegraphics[width=5.2cm]{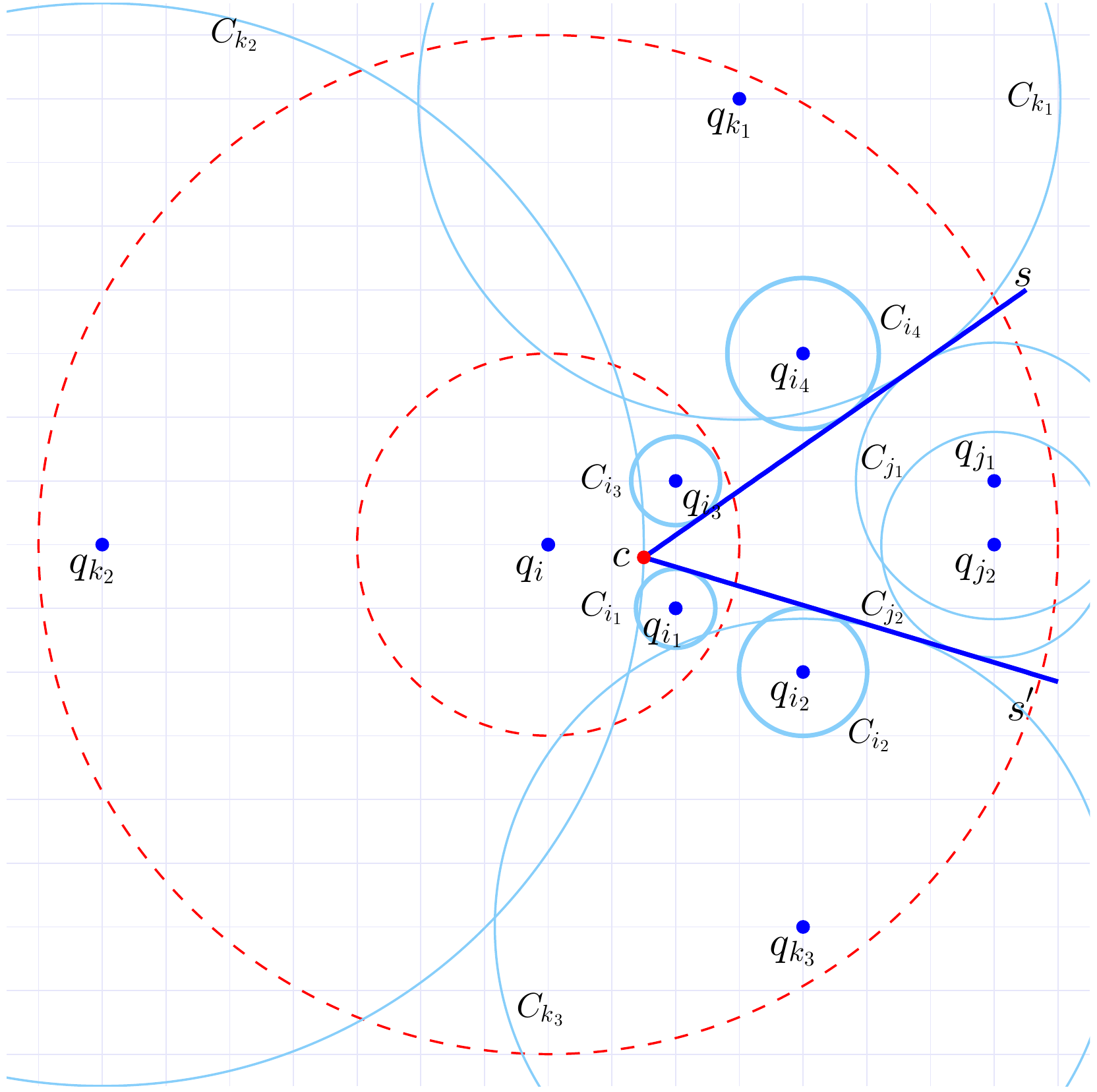}
	
	\vspace{-0.1in}
	\caption{ Left:  $\{c\}=C_{i_1}\cap C_{i_2}\cap C_{i_3}$ and $B_{i_1}\subset B_{i_2}\cup B_{i_3}$.
		      Center:  the angle between $s$ and $s'$ is at most $\frac{\pi}{4}$ and
		      $\{c\}=C_{i_1}\cap C_{i_2}\cap C_{i_3}$ and $B_{i_1}\subset B_{i_2}\cup B_{i_3}$.
		      Right: $C_{i_1}$, $C_{i_2}$ are tangent to $s$, and $C_{i_3}$, $C_{i_4}$ are tangent to $s'$,
		      For each one of these four circles, any tangent line segment, except $s, s'$, cannot  be extended outside $B_{i,8\eta}$ without intersecting with any other circle. }
	\label{fig find a criticall point}
	
\end{figure}

\begin{lemma}  \label{lemma find a criticall point}
Suppose $c'\in B_{i,3\eta}$ is a critical point of $\gamma$, and $(c,S)$ is the output of \textsc{FindCritical}$(q_i)$ (Algorithm \ref{find a criticall point}), then $c=c'$. Moreover, $|S|=1$ if and only if $c$ is an endpoint of $\gamma$.
\end{lemma}

\begin{proof}
Since $c'\in B_{i,3\eta}$ and $\eta<\frac{\tau}{32}$, we have $B_{i,8\eta}\subset B(c',\frac{\tau}{2})$.
So, from (R2) we know in $B_{i,8\eta}$, $\gamma$ either has one line segment which means $c'$ is an endpoint, or has two line segments which means $c'$ is an internal critical point.

If $c'$ is an endpoint, then the line segment of $\gamma$ must satisfy Condition \textsc{EndPoint} in Algorithm \ref{find a criticall point}.
Moreover, if in Algorithm \ref{find a criticall point} $s$ satisfies Condition \textsc{EndPoint}, then $c$ must be a critical point of $\gamma$.
This is because, as show in Figure \ref{fig find a criticall point}(Left), there exists three points $q_{i_1}, q_{i_2},q_{i_3}\in Q_{i,8\eta}$ such that $\{c\}=C_{i_1}\cap C_{i_2}\cap C_{i_3}$ and $B_{i_1}\subset B_{i_2}\cup B_{i_3}$ and the tangent of $C_{i_1}$ at $c$ intersects with $\dot{B}_{i_2}\cup \dot{B}_{i_3}$.
This can be seen by observing there must exists points $q_{i_2}, q_{i_3} \in Q_{i,8\eta}$ which are
 (i) on the opposite side from  $s$ of the perpendicular to $s$ through $c$,
 (ii) are a distance at least $3 \eta$ from $c$, and
 (iii) within a distance of $3 \eta$ from each other.
This implies there exists another point $q_{i_1} \in Q \cap B_{i_2} \cap B_{i_3}$ and with $v_{i_1} \leq 2 \eta$.  Hence $B_{i_1}$ must be contained in $B_{i_2} \cup B_{i_3}$.
Thus,  (O3) implies $c$ is a critical point of $\gamma$, and from (O1) we know $c=c'$.

If $c'$ is an interior point, then as show in Figure \ref{fig find a criticall point}(Center and Right), no line segment can satisfy Condition \textsc{EndPoint} in Algorithm \ref{find a criticall point}, so the algorithm will not stop before the third loop.
Then the two line segments of $\gamma$ with $c'$ as the common endpoint can satisfy Condition \textsc{InternalPoint}.
Moreover, if $s$ and $s'$ satisfy Condition \textsc{InternalPoint}, then we will show $c$ must be a critical point of $\gamma$.
There are two possibilities: the angle between $s$ and $s$ at most $\frac{\pi}{4}$, or greater than $\frac{\pi}{4}$.

If the angle is less than or equal to $\frac{\pi}{4}$, then as shown in Figure \ref{fig find a criticall point}(Center),
there exists three points $q_{i_1}, q_{i_2},q_{i_3}\in Q_{i,8\eta}$ such that $\{c\}=C_{i_1}\cap C_{i_2}\cap C_{i_3}$ and $B_{i_1}\subset B_{i_2}\cup B_{i_3}$ and the tangent of $C_{i_1}$ at $c$ intersects with $\dot{B}_{i_2}\cup \dot{B}_{i_3}$.  This follows by the same argument as when $c$ is an endpoint.
So,  (O3) implies $c$ is a critical point of $\gamma$, and from (O1) we know $c=c'$.

If the angle is greater than $\frac{\pi}{4}$, then as shown in Figure \ref{fig find a criticall point}(Right), there exists four points $q_{i_1}, q_{i_2},q_{i_3}, q_{i_4}\in Q_{i,8\eta}$ outside the interior angular region, and such that
$C_{i_1}$, $C_{i_2}$ are tangent to $s'$, and $C_{i_3}$, $C_{i_4}$ are tangent to $s$.
Moreover, these four circles can be chosen to not intersect with each other.
Next we can argue that because the angle is sufficiently large, we can block a path from $c'$ to outside of $B_{i,8\eta}$ both inside the interior angular region, and outside it.
Outside this region, we can choose three points in $q_{k_1}, q_{k_2}, q_{k_3} \in Q_{i,8\eta}$ of which $C_{k_1}$ is incident to $\mathsf{ray}(s)$, $C_{k_2}$ is incident to $c'$, and $C_{k_3}$ is incident to $\mathsf{ray}(s')$; and that $\dot{B}_{k_1}$ and $\dot B_{k_2}$ intersect and $\dot B_{k_2}$ and $\dot B_{k_3}$ intersect.
Similarly, inside the interior angular region, we can chose two points $q_{j_1}, q_{j_2} \in Q_{i,8\eta}$ so $C_{j_1}$ and $C_{j_2}$ are incident to $\mathsf{ray}(s)$ and $\mathsf{ray}(s')$, respectively, and that $\dot B_{j_1}$ and $\dot B_{j_2}$ intersect.
These two sets of points blocks any other straight path from $c'$ (required by (O2)) from existing $B_{i,8\eta}$ (required by (O1)) without entering the interior of some $\dot B_k$.  And the first four points $q_{i_1}, q_{i_2},q_{i_3}, q_{i_4}$ ensures that this $c'$ is unique (by (O1)) and $c' = c$ must be a critical point on $\gamma$.
\end{proof}


Using 	Algorithm \ref{the existence of critical point} and \ref{find a criticall point}
we can find all critical points $(E,A)$ with associated line segments of $\gamma$, so the final step is to use function \textsc{DetermineOrder}$(E,A)$ (Algorithm \ref{alg:DOC}) to determine their order, as we argue it will completely recover $\gamma$.


\begin{algorithm}
	\caption{\label{alg:DOC} \textsc{DetermineOrder}$(E,A)$: Determine the order of critical points}
	\begin{algorithmic}

		\STATE Choose any $(c_0,S_0)\in E$, let $k=|A|-1$, $A:=A\setminus \{(c_0,S_0)\}$, $s_1\in S_0$ and $\gamma:=\langle c_0\rangle$.
		\FOR {$i=1$ \textbf{to} $k$}
 		  \STATE  Find closest $c$ from $(c,S) \in A$ to $c_{i-1}$ such that $c$ is on $\mathsf{ray}(s_{i})$,
                  and let $A := A \setminus \{(c, S)\}$.
		  \STATE  Append $c_i = c$ to $\gamma$, and if $i<k$ then let $s_{i+1}= s$ where  $s \in S$ is not parallel with $s_{i}$.
		\ENDFOR
		\RETURN $\gamma$
	\end{algorithmic}
\end{algorithm}

\begin{theorem} 
	Suppose $Q = G_\eta \cap \Omega$, $\eta \leq \frac{\tau}{32}$, and $v_Q(\gamma)$ is generated by $Q$ and $\gamma\in \T_\tau(\Omega)$, then Algorithm \ref{alg:full-recover} can recover $\gamma$ from $v_Q(\gamma)$ in $O(|Q|+k^2)$ time, where $k$ is the number of line segments of $\gamma$.
\end{theorem}

\begin{proof}
From Lemmas \ref{the existence of critical point} and \ref{find a criticall point} we know Algorithms \ref{the existence of critical point} and \ref{find a criticall point}  identify all critical points of $\gamma$, and the line segments of $\gamma$ associated with each critical point.
So we only need to show Algorithm \ref{alg:DOC} determines the correct order of critical points.
This is because if a point moves along $\gamma$ it cannot stop or change direction until it hits a critical point (Observation (O2)), and when it hits a critical point it has to stop or change direction, otherwise it will violate (R1) or (R2).
So, Algorithm \ref{alg:DOC} starts from an endpoint and moves along the direction of line segment associated with it, and changes the direction only after arriving at the next critical point, until all critical points are visited.
This gives the correct order of critical points of $\gamma$.

Moreover, the running time of Algorithm \ref{the existence of critical point} and \ref{find a criticall point} are constant, since they both only examine a constant number of points, circles, etc in each $B_{i,3\eta}$ or $B_{i,8\eta}$.  And these can be retrieved using the implicit grid structure in constant time.
Thus the \textbf{for} loop in Algorithm \ref{alg:full-recover} takes $O(|Q|)$ time.
%
The final Algorithm \ref{alg:DOC} to recover the order takes $O(k^2)$ time, since a constant fraction of steps need to check a constant fraction of all critical points in $A$.  So, the total running time of this algorithm is $O(|Q|+k^2)$.
\end{proof}

\section{Conclusion}
\label{sec:conclude}

In this paper we analyze sketches via the $v_i(J)  = \mindist(q_i,J) = \inf_{p \in J} \|p-q_i\|$ procedure, for a variety of geometric objects, and show how many and how measurement points $Q$ can be chosen.  Collecting $n$ values $v_Q(J) = (v_1(J), \ldots, v_n(J))$ leads to a simple to use and natural distance $\dQ(J_1, J_2) = \|v_Q(J_1) - v_Q(J_2)\|$.

For hyperplanes, the sensitivity sampling framework can be applied fairly directly to chose $Q$, and requires about $(d/\eps^2)$ points or $O(d^2 \log d / \eps^2)$ for stronger guarantees.
However, for more general objects we show that a resolution parameter $L/\rho$ needs to be introduced, and affects the sample size even in $\R^2$.  For instance, when the goal is to represent shapes by their \mindist function as defined over a domain $[0,L]^2$, then $\Theta(\frac{L}{\rho} \frac{1}{\eps^2})$ samples are required for $\eps$-error.
For the case of piecewise-linear curves (e.g., trajectories) we can provide even stronger error guarantees.  By bounding an associated shattering dimension for curves of length at most $k$, we can provide strong approximation guarantees on $\dQ$ using roughly $\frac{L}{\rho} \log^2 \frac{L}{\rho} \cdot k^3 \cdot \frac{1}{\eps^2}$ samples.  Moreover, we can exactly recover the trajectory $\gamma$ using only its \mindist sketched vector $v_Q(\gamma)$.

While a companion paper~\cite{PT19a} has provided experimental results which demonstrate this distance $\dQ$ is convenient and powerful in trajectory classification tasks, many other open questions remain.  These include extending similar representations to other tasks, in theory and in practice.  Moreover, our bounds rely on a few related minimum resolution parameters
  $\rho \leq \min_{\gamma_1,\gamma_2} \dQ(\gamma_1, \gamma_2)$,
  $\tau$ is a gap between points in a grid $Q$ required to recover $\gamma$, and
  $\eta = \min_{q,q' \in Q \;\; q\neq q'} \|q-q'\|$ is a pairwise minimum distance on the multiset $Q$.
While these parameters are in all scenarios we considered asymptotically equivalent, it would be useful to unify these terms in a single theory.
Finally, we would like to show not just exact trajectory recovery (under conditions on $\gamma$ and $Q$), but also to loosen those restrictions and provided topological recovery (e.g., geometric inference~\cite{GeomInf}) conditions for boundaries of compact sets.

\newpage
\bibliographystyle{abbrv}

\normalsize

\newpage
\appendix

\section{A Lemma Used in Section \ref{Approximate the Distance Between Trajectories}}  \label{proof of theorem sensitivity}

\begin{lemma} \label{bound the dimension of range space intersection and union}
	\label{the dimeison union intersection}
	Suppose $Q\subset \R^d, X_1\subset \R^{d_1}, X_2\subset \R^{d_2}$, and
	$\c{R}_1=\{ \{q\in Q|\ g_1(q,x)\leq 0\} |\ x\in X_1\}$,
	$\c{R}_2=\{ \{q\in \R^2|\ g_2(q,x)\leq 0\}|\  x\in X_2\}$ where $g_1, g_2$ can be any fixed real functions. Define
	$\c{R}_3=\{\{q\in \R^2|\ g_1(q,x_1)\leq 0 \} \cap \{q\in \R^2|\ g_2(q,x_2)\leq 0\}|\ x_1\in X_1,x_2\in X_2\}$,
	$\c{R}_4=\{\{q\in \R^2|\ g_1(q,x_1)\leq 0 \} \cup \{q\in \R^2|\ g_2(q,x_2)\leq 0\}|\ x_1\in X_1,x_2\in X_2\}$.
	If $\text{dim}(\R^2,\c{R}_1)=s_1$ and $\text{dim}(\R^2,\c{R}_2)=s_2$, then $\text{dim}(\R^2,\c{R}_3)\leq s_1+s_2$
	and $\text{dim}(\R^2,\c{R}_4)\leq s_1+s_2$.
\end{lemma}

\begin{proof}
	Suppose $G\subset \R^2$ and $|G|\leq \infty$, then we have
	\begin{equation}
	\{G\cap R|\ R\in \c{R}_3\}= \{(G\cap R_1)\cap(G\cap R_2)|\ R_1\in \c{R}_1, R_2\in \c{R}_2  \}.
	\end{equation}
	So, we have
	\begin{equation}
	\begin{split}
	&|\{G\cap R|\ R\in \c{R}_3\}|= |\{(G\cap R_1)\cap(G\cap R_2)|\ R_1\in \c{R}_1, R_2\in \c{R}_2  \}|\\
	\leq& |\{G\cap R_1|\ R_1\in \c{R}_1\}|\times|\{G\cap R_2|\ R_2\in \c{R}_2\}|\leq |G|^{s_1}|G|^{s_2}
	=|G|^{s_1+s_2}.
	\end{split}
	\end{equation}
	which implies  $\text{dim}(\R^2,\c{R}_3)\leq s_1+s_2$, and similarly we have $\text{dim}(\R^2,\c{R}_4)\leq s_1+s_2$.	
\end{proof}

\section{Sensitivity Computation and its Relationship with Leverage Score}
\label{sec:levage}

In this section, we describe how to compute the sensitivity score $\sigma(x_i)$ for each $x_i \in Q$.
To this end, we can invoke a theorem about vector norms by Langberg and Shulman~\cite{MLLJ2010}:

\begin{lemma}[Theorem 2.2 in ~\cite{MLLJ2010}]
	\label{lemma sensitivity}
	Suppose $\mu$ is a probability measure on a metric space $X$, and $V=\{v:X\mapsto \R\}$ is a real vector space of dimension $\kappa$. Let $F=\{f: X\mapsto [0,\infty)\ | \ \exists\ v\in V \text{ s.t. } f(x)=v(x)^2,\ \forall x\in X \}$, and $\{v^{(1)},\cdots,v^{(\kappa)}\}$ be an orthonormal basis for $V$ under the inner product $\langle u,v \rangle:=\int _X u(x)v(x) \dir\mu(x)$, $\forall u,v \in V$. Then, $\sigma_{F,X,\mu}(x)=\sum _{i=1}^\kappa v^{(i)}(x)^2$ and $\mathfrak{S}(F)=\kappa$.
\end{lemma}

We have already set $X=Q$ and $\mu = \frac{1}{n}$, and have defined $V$ and $F$.  To apply the above theorem need to define an orthonormal basis $\{v^{(1)}, v^{(2)}, \ldots, v^{(d+1)}\}$ for $V$.  A straightforward basis (although not necessarily an orthonormal one) exists as
$v^{(d+1)}(q)=v_{e^{(d+1)}}(q)=1$
and
$v^{(i)}(q)=v_{e^{(i)}}(q)=x_i$ for all $i\in[d]$ and $q=(x_1,\cdots,x_d)\in \R^d$, where $e^{(i)}=(0,\cdots,0,1,0,\cdots,0)$ is an indicator vector with all zeros except 1 in $i$th coordinate.
That is the $i$th basis element $v^{(i)}$ is simply the $i$th coordinate of the input.
Since $Q$ is full rank,  $\{v^{(1)},\cdots,v^{(d+1)}\}$ is a basis of $V$.

We are now ready to state our theorem on computing sensitivity scores on a general $(F,Q,\mu)$, where we typically set $\mu = \frac{1}{n}$.


\begin{theorem}\label{theorem sensitivity}
	Suppose $\mu$ is a probability measure on a metric space $Q=\{q_1,\cdots,q_n\}$ such that $\mu(q_i)=p_i>0$ for all $i\in [n]$,
	$V=\{v: Q\mapsto \R\}$ is a real vector space of dimension $\kappa$ with a basis  $\{v^{(1)},\cdots,v^{(\kappa)}\}$, and  $F=\{f: Q\mapsto [0,\infty)\ | \ \exists\ v\in V \text{ s.t. } f(q)=v(q)^2,\ \forall q\in Q \}$.
	If we introduce a $\kappa\times n$ matrix $A$ whose $i$th column $a_{i}$ is defined as:
	$a_{i}=(v^{(1)}(q_i)\sqrt{p_i},\cdots,v^{(\kappa)}(q_i)\sqrt{p_i})^T$,
	then we have
	\begin{equation}   \label{sensitivity and leverage score}
	\sigma_{F,Q,\mu}(q_i) \cdot p_i=a_{i}^T(AA^T)^{-1}a_{i},\ \ \ \forall\ q_i\in Q.
	\end{equation}
\end{theorem}

\begin{proof}
	Suppose the QR decomposition of $A^T$ is $A^T=\tilde{Q}\tilde{R}$, where $\tilde{Q}$
	is an $n\times \kappa$ orthogonal matrix ($\tilde{Q}^T\tilde{Q}=I$), and $\tilde{R}$ is an $n\times n$  upper triangular matrix.
	Since $\{v^{(1)},\cdots,v^{(\kappa)}\}$ is a basis of $V$, the columns of $A^T$ are linear independent, which implies the matrix
	$\tilde{R}$ is invertible. Using the fact that $\tilde{Q}^T\tilde{Q}$ is an identity matrix, we have
	\begin{equation} \label{the relationship between A and Q}
	\begin{split}
	A^T(AA^T)^{-1}A=&\tilde{Q}\tilde{R}(\tilde{R}^T\tilde{Q}^T\tilde{Q}\tilde{R})^{-1}\tilde{R}^T\tilde{Q}^T
	=\tilde{Q}\tilde{R}(\tilde{R}^T\tilde{R})^{-1}\tilde{R}^T\tilde{Q}^T\\
	=&\tilde{Q}\tilde{R}\tilde{R}^{-1}(\tilde{R}^T)^{-1}\tilde{R}^T\tilde{Q}^T
	=\tilde{Q}\tilde{Q}^T
	\end{split}
	\end{equation}
	From Lemma \ref{lemma sensitivity}, we have $\sigma_{F,Q,\mu}(q_i)=\sum_{j=1}^\kappa(\tilde{Q}_{i,j})^2$, which is the $i$-th
	entry on the diagonal of $\tilde{Q}\tilde{Q}^T$, so from \eqref{the relationship between A and Q}, we obtain \eqref{sensitivity and leverage score}.
\end{proof}


%
This theorem not only shows how to compute the sensitivity of a point, but also gives the relationship between sensitivity and the leverage score.

\paragraph*{Leverage score.}
Let $(\cdot)^+$ denotes the Moore-Penrose pseudoinverse of a matrix, so $(AA^T)^+=(AA^T)^{-1}$ when $AA^T$ is full rank.
The \textit{leverage score}~\cite{DMM08} of the $i$th column $a_i$ of matrix $A$ is defined as:
$\tau_i(A):= a_i^T(AA^T)^+a_i.$

This definition is more specific and linear-algebraic than sensitivity.  However, Theorem \ref{theorem sensitivity} shows that value $\sigma_{F,Q,\mu}(x_i) \cdot p_i$ is just the leverage score of the $i$th column of the matrix $A$.
Compared to sensitivity, leverage scores have received more attention for scalable algorithm development and approximation~\cite{DMM08,boutsidis2009improved,DMMW12,CMM17,MM17,MCJ2016}

\subsection{Estimate the Distance by Online Row Sampling}
\label{sec: online row sampling}

If the dimensionality is too high and the number of points is too large to be stored and processed in memory, we can apply online row sampling~\cite{MCJ2016} to estimate $\dQ$.
Note that as more rows are witnessed the leverage score of older rows change.  While other approaches (c.f. \cite{DMMW12,CMM17,MM17}) can obtain similar (and maybe slightly stronger) bounds, they rely on more complex procedures to manage these updating scores.
The following Algorithm \ref{Alg: online row sampling} by Cohen \etal~\cite{MCJ2016}, on the other hand, simply samples columns as they come proportional to their estimated ridge leverage score~\cite{CMM17}; thus it seems like the ``right'' approach.


\begin{algorithm}[H]
	\caption{\textsc{Online-Sample}$(A,\eps,\delta)$}  \label{Alg: online row sampling}
	\begin{algorithmic}
		\STATE Set $\lambda:=\frac{\delta}{\varepsilon}$, $c:=8 \log(\frac{d}{\varepsilon^2})$, and let $\widetilde{A}$ be empty (a $0\times d$ matrix).
		\FOR {rows $a_i \in A$}
		\STATE Let $p_i:=\min(c \cdot (1+\eps)a_i^T(\widetilde{A}^T\widetilde{A}+\lambda I)^{-1}a_i,1)$.
		\STATE With probability $p_i$, append row $a_i / \sqrt{p_i}$ to $\widetilde{A}$; otherwise do nothing.
		
		\ENDFOR
		\RETURN $\widetilde{A}$.
	\end{algorithmic}
\end{algorithm}

According to the Theorem 3 in \cite{MCJ2016}, Algorithm \ref{Alg: online row sampling} returns a matrix $\widetilde{A}$, with high probability, such that
$(1-\eps)A^TA-\delta I \preceq \widetilde{A}^T\widetilde{A} \preceq(1+\eps)A^TA+\delta I$, and the number of rows in $\widetilde{A}$ is $O(d\log(d)\log(\eps\|A\|_2^2/\delta)/\eps^2)$.
(Recall $A\preceq B$ means $x^TAx\leq x^TBx$ for every vector $x$.)

Given a set of points $Q=\{q_1,\cdots,q_n\}\subset \R^d$, where $q_i$ has the coordinates $(x_{i,1},\cdots,x_{i,d})$, we introduce an $n \times (d+1)$ matrix $A_Q$ whose
$i$th row $a_{i}$ is defined as:
\[
a_{i}=(x_{i,1},\cdots,x_{i,d},1),
\]
For any two hyperplanes $h_1,h_2$, they can be uniquely expressed by vectors $u^{(1)}, u^{(2)} \in \b{U}^{d+1}$, and define $u=u^{(1)}-u^{(2)}\in \R^{d+1}$, then we have $\dQ(h_1,h_2)=\frac{1}{\sqrt{n}}\|A_Qu\|$. So, if $n$ is very large we can apply
Algorithm \ref{Alg: online row sampling} to efficiently
sample rows from $A_Q$, and use $A_{\tilde{Q}}$ to estimate $\dQ(h_1,h_2)$.
From  Theorem 3 in  \cite{MCJ2016}, we have the following result.

\begin{theorem}\label{theorem estimate distance online}
	Suppose a set $Q$ and matrix $A_Q$ are defined as above.
	Let $A_{\tilde{Q}}=\text{Online-Sample}(A_Q,\varepsilon,\delta)$ be the matrix returned by Algorithm \ref{Alg: online row sampling}.   Then, with probability at least $1-\frac{1}{d+1}$, for any two hyperplanes $h_1,h_2$ expressed by $u^{(1)},u^{(2)} \in \b{U}^{d+1}$, suppose $u_{h_1,h_2}=u^{(1)}-u^{(2)}$, we have
	\[ 
	\frac{1}{1+\eps}\big(\frac{1}{n}\|A_{\tilde{Q}} u_{h_1,h_2}\|^2-\frac{1}{n}\delta\|u_{h_1,h_2}\|^2\big)^\frac{1}{2}
	\leq\mathtt{d}_Q(h_1,h_2)
	\leq \frac{1}{1-\eps}\big(\frac{1}{n}\|A_{\tilde{Q}} u_{h_1,h_2}\|^2+\frac{1}{n}\delta\|u_{h_1,h_2}\|^2\big)^\frac{1}{2},
	\] 
	where
	$\|\cdot\|$ is the Euclidean norm, and with probability at least $1-\frac{1}{d+1}-e^{-(d+1)}$
	the number of rows in $A_{\tilde{Q}}$ is $O(d\log(d)\log(\varepsilon\|A_Q\|_2^2/\delta)/\varepsilon^2)$.
\end{theorem}

To make the above bound 
hold with arbitrarily high probability, we can use the standard median trick: run Algorithm \ref{Alg: online row sampling} $k$ times in parallel to obtain $A_{\tilde{Q}_1},\cdots,A_{\tilde{Q}_k}$, then for any two hyperplanes  $h_1,h_2$, we take the median of $\|A_{\tilde{Q}_1} u_{h_1,h_2}\|^2,\cdots,\|A_{\tilde{Q}_k} u_{h_1,h_2}\|^2$.

\paragraph*{Remark.}
Since $u_{h_1,h_2}=u^{(1)}-u^{(2)}$, we have
\begin{equation*}
\begin{split}
\|u_{h_1,h_2}\|^2=&(\|u^{(1)}-u^{(2)}\|)^2\leq(\|u^{(1)}\|+\|u^{(2)}\|)^2\leq2(\|u^{(1)}\|^2+\|u^{(2)}\|^2)\\
=&2\big(2+(u^{(1)}_{d+1})^2+(u^{(2)}_{d+1})^2\big)=4+2\mathtt{d}^2(\mathbf{0},h_1)+2\mathtt{d}^2(\mathbf{0},h_2),
\end{split}
\end{equation*}
where $\mathtt{d}(\mathbf{0},h)$ is the distance from a choice of origin $\mathbf{0}$ to $h$.
If we assume that any hyperplanes we consider must pass within a distance $\Delta$ to the choice of origin, then let $\Delta' = 4(1+\Delta^2)$ and $\|u_{h_1,h_2}\|^2 \leq \Delta'$.
Now $\mathtt{d}_{\tilde{Q}, W}(h_1,h_2))^2=\frac{1}{n}\|A_{\tilde{Q}}u_{h_1,h_2}\|^2$ where $\tilde{Q}$ is the  set of points corresponding to rows in $A_{\tilde{Q}}$, and the weighting $W$ is defined so $w_i = |\tilde{Q}| / n$.
Then the conclusion of Theorem~\ref{theorem estimate distance online} can be
rewritten as
\[
\frac{1}{1+\eps}\big(\mathtt{d}_{\tilde{Q},W}(h_1,h_2)^2 - \frac{\Delta' \delta}{n}  \big)^\frac{1}{2}
\leq\mathtt{d}_Q(h_1,h_2)
\leq \frac{1}{1-\eps}\big(\mathtt{d}_{\tilde{Q},W}(h_1,h_2)^2 + \frac{\Delta' \delta}{n} \big)^\frac{1}{2},
\]
which means $\mathtt{d}_Q(h_1,h_2)$ can be estimated by $\mathtt{d}_{\tilde{Q},W}(h_1,h_2)$ and the bound $\Delta$ on the distance to the origin.  Recall the distance and the bound in Theorem \ref{theorem estimate distance online} is invariant to the choice of $\mathbf{0}$, so for this interpretation it can always be considered so $\Delta$ is small.

\end{document}